\documentclass{amsproc}

\usepackage{latexsym}
\usepackage{amsmath}
\usepackage{bm}
\usepackage{amssymb}
\usepackage{amsfonts}
\usepackage{amsthm}
\usepackage{graphicx}
\usepackage{subfig}
\usepackage{tikz,tikz-qtree,tikz-qtree-compat}
\usetikzlibrary{shapes,snakes,arrows,calc,fit}
\usepackage{color,xspace,colortbl}
\usepackage{float}
\usepackage{algorithmicx}
\usepackage{algorithm}
\usepackage[noend]{algpseudocode}
\usepackage{enumitem}
\usepackage{balance}
\usepackage{todonotes}
\usepackage{url}
\usepackage{multirow}

\colorlet{LightViolet}{violet!40}
\colorlet{LightRed}{red!40}
\colorlet{LightOrange}{orange!40}
\colorlet{LightGreen}{green!40}
\colorlet{LightBlue}{blue!40}
\colorlet{DarkGreen}{green!50!black}
\colorlet{DarkRed}{red!70!black}
\colorlet{DarkCyan}{red!70!black}
\colorlet{DarkBlue}{blue!80!black}
{\definecolor{DarkOrange}{rgb}{1.0, 0.49, 0.0}
\definecolor{Airforceblue}{rgb}{0.36, 0.54, 0.66}

\newcommand{\nop}[1]{}

\newcommand{\dc}{\text{\sf DC}}

\newcommand{\incomp}{\perp}
\newcommand{\flow}{{\sf excess}}

\newcommand{\Dsimple}{D_{\bm \delta}^{\sf simple}}
\newcommand{\Dflow}{D_{\bm \delta}^{\sf flow}}

\newcommand{\flowbound}{\text{\sf flow-bound}}

\newcommand{\calL}{\mathcal L}

\newcommand{\calP}{\mathcal P}

\newcommand{\calF}{\mathcal F}

\newcommand{\R}{\mathbb R} 
\newcommand{\Mod}{\text{\sf M}}
\newcommand{\Nor}{\text{\sf N}}

\newcommand{\DC}{\mathrm{DC}}

\newcommand{\be}{\begin{enumerate}}
\newcommand{\ee}{\end{enumerate}}
\newcommand{\bi}{\begin{itemize}}
\newcommand{\ei}{\end{itemize}}
\newcommand{\beq}{\begin{equation}}
\newcommand{\eeq}{\end{equation}}

\newcommand{\bp}{\begin{proof}}
\newcommand{\ep}{\end{proof}}
\newcommand{\bcor}{\begin{cor}}
\newcommand{\ecor}{\end{cor}}
\newcommand{\bthm}{\begin{thm}}
\newcommand{\ethm}{\end{thm}}
\newcommand{\blmm}{\begin{lmm}}
\newcommand{\elmm}{\end{lmm}}
\newcommand{\bdefn}{\begin{defn}}
\newcommand{\edefn}{\end{defn}}
\newcommand{\bprop}{\begin{prop}}
\newcommand{\eprop}{\end{prop}}
\newcommand{\bconj}{\begin{conj}}
\newcommand{\econj}{\end{conj}}
\newcommand{\bopm}{\begin{opm}}
\newcommand{\eopm}{\end{opm}}
\newcommand{\brmk}{\begin{rmk}}
\newcommand{\ermk}{\end{rmk}}

\newcommand{\suchthat}{\ | \ }

\theoremstyle{plain}                   
\newtheorem{thm}{Theorem}[section]
\newtheorem{lmm}[thm]{Lemma}
\newtheorem{prop}[thm]{Proposition}
\newtheorem{cor}[thm]{Corollary}

\theoremstyle{definition}              

\newtheorem{opm}{Open Problem}
\newtheorem{conj}{Conjecture}
\newtheorem{ex}{Example}

\newtheorem{defn}{Definition}

\newtheorem{rmk}{Remark}
\newtheorem{claim}{Claim}

\definecolor{Red}{RGB}{255,204,204}
\definecolor{Green}{RGB}{204,255,204}
\definecolor{Blue}{RGB}{204,204,255}

\usepackage{xpatch}
\usepackage{textcase}
\makeatletter
\xpatchcmd{\@sect}{\uppercase}{\MakeTextUppercase}{}{}
\xpatchcmd{\@sect}{\uppercase}{\MakeTextUppercase}{}{}
\makeatother

\allowdisplaybreaks[1]

\usepackage{booktabs} 

\usepackage{fullpage}

\author{Sungjin Im}
\address{University of California, Merced, CA 95343}

\author{Benjamin Moseley}
\address{Tepper School of Business, Carnegie Mellon University, Pittsburgh, PA 15213}

\author{Hung Q. Ngo}
\address{RelationalAI, Inc., Berkeley, CA 94704}

\author{Kirk Pruhs}
\address{Computer Science Department, University of Pittsburgh, Pittsburgh, PA 15260}

\begin{document}

\title{Efficient Algorithms for Cardinality Estimation and Conjunctive Query Evaluation
With Simple Degree Constraints}
\renewcommand{\shorttitle}{Efficient Cardinality Estimation with Simple Degree Constraints}


\begin{abstract}
Cardinality estimation and conjunctive query evaluation are two of the most fundamental problems in database query processing. Recent work proposed, studied, and implemented a robust and practical information-theoretic cardinality estimation framework. In this
framework, the estimator is the cardinality upper bound of a conjunctive query subject to ``degree-constraints'', which model a rich set of input data statistics. For general degree constraints, computing this bound is computationally hard. Researchers have naturally sought efficiently computable
relaxed upper bounds that are as tight as possible. The
polymatroid bound is the tightest among those relaxed upper bounds. While it is an open question whether the polymatroid bound can be computed in polynomial-time in general, it is known to be computable in polynomial-time for some classes of degree constraints.

Our focus is on a common class of degree constraints called simple degree constraints. Researchers had not previously determined how to compute the polymatroid bound in polynomial time for this class of constraints.
Our first main result is a polynomial time algorithm to compute the polymatroid bound given simple degree constraints. Our second main result is a polynomial-time algorithm to compute a ``proof sequence'' establishing this bound. This proof sequence can then be incorporated in the PANDA-framework to give a faster algorithm to evaluate a conjunctive query. In addition,
we show computational limitations to extending our results to broader classes of degree constraints. Finally, our technique leads naturally to a new relaxed upper bound called the {\em flow bound}, which is computationally tractable.
\end{abstract}

\keywords{Cardinality estimation, conjunctive query evaluation, polymatroid bound, proof sequence, polynomial time algorithms}

\maketitle

\section{Introduction}
\label{sec:intro}

\subsection{Motivations}
\label{subsec:motivations}

Estimating a tight upper bound on the output size of a (full) conjunctive query is an
important problem in query optimization and evaluation, for many reasons. First, the bound
is used to estimate whether the computation can fit within a given memory budget. Second,
cardinalities of intermediate relations are the main parameters used to estimate the cost of
query plans~\cite{lohman}, and the intermediate size bound is used as a {\em cardinality
estimator}~\cite{DBLP:conf/icdt/000122, DBLP:journals/pacmmod/DeedsSB23,
DBLP:conf/icdt/DeedsSBC23, DBLP:journals/pacmmod/KhamisNOS24,
DBLP:journals/sigmod/KhamisDOS24}. These ``pessimistic'' estimators are designed to be
robust~\cite{DBLP:journals/corr/abs-2502-05912}, removing the assumptions that lead to the
well-documented impediment of selectivity-based
estimators~\cite{DBLP:journals/pvldb/LeisGMBK015}, where the estimates often under
approximate the real intermediate sizes by orders of magnitudes. Third, the bound is used as
a yardstick to measure the quality of a join algorithm~\cite{DBLP:conf/pods/000118}.

The history of bounding the output size of a conjunctive query (CQ) is
rich~\cite{MR1338683,MR0031538,MR859293,MR599482,MR1639767,MR2104047,suciu2023applications,
GLVV,panda,csma,DBLP:conf/pods/000118,DBLP:journals/talg/GroheM14,AGM,DBLP:journals/corr/abs-2402-02001}.
In the simplest setting where the query is a join of base tables with known cardinalities,
the {\em AGM-bound}~\cite{AGM} is tight w.r.t. data complexity. up to a $O(2^n)$ factor,
where $n$ is the number of variables in the query. Join algorithms running in asymptotically
the same amount of time are ``worst-case optimal''~\cite{NPRR,LFTJ,skew}. In practice,
however, we know a lot more information about the inputs beyond cardinalities. Query
optimizers typically make use of distinct value counts, heavy hitters, primary and and
foreign keys information~\cite{looking-glass,DBLP:conf/icdt/000122}, in addition to function
predicates. Every additional piece of information is a constraint that may drastically
reduce the size of the query output. For example, in a triangle query $R(a,b)\wedge
S(b,c)\wedge T(a,c)$, if we knew $a$ was a primary key in relation $R$, then the output size
is bounded by $\min\{|T|, |R|\cdot |S|\}$, which can be asymptotically much less than the
AGM-bound of $\sqrt{|R|\cdot |S| \cdot |T|}$.

To model these common classes of input statistics and key constraints, Abo Khamis et
al.~\cite{csma,panda} introduced an abstract class of constraints, called degree
constraints.\footnote{Degree constraints are also generalized to frequency moments collected
on input tables~\cite{DBLP:journals/corr/abs-2402-02001,DBLP:conf/icdt/000122}.} A degree
constraint is a triple $(X,Y,c)$, where $X \subseteq Y$ are sets of variables, and $2^c$ is
positive integer. The constraint holds for a predicate or relation $R$ if $\max_{\bm x}
|\pi_Y \sigma_{X=\bm x} R| \leq 2^c$, where $\pi$ and $\sigma$ are the projection and
selection operators, respectively. The intuitive meaning of the triple $(X, Y, c)$ is that
the number of possible bindings for variables in $Y$ can attain, given a particular binding
of the variables in $X$, is at most $2^c$. If $X=\emptyset$ then this is a {\em cardinality
constraint}, which says $|\pi_YR|\leq 2^c$; if $c=0$ (i.e. $2^c=1$) then this is a {\em
functional dependency} from $X$ to $Y$.

Given a set $\dc$ of degree constraints, and a full CQ $Q$; we write $\bm I \models \dc$ to
denote  that a database instance $\bm I$ satisfies the constraints $\dc$. A robust
cardinality estimator is $\sup_{\bm I \models \dc} |Q(\bm I)|$, the upper bound on the
number of output tuples of $Q$ over all databases satisfying the
constraints~\cite{DBLP:conf/icdt/000122,
DBLP:conf/sigmod/CaiBS19,
DBLP:journals/sigmod/KhamisDOS24}.
We will refer to this bound as the {\em combinatorial bound}. For arbitrary degree
constraints, it is not clear how to even compute it. Even in the simplest case when all
degree constraints are cardinality constraints, approximating the combinatorial bound better
than the $2^n$-ratio is already hard~\cite{AGM}, where $n$ is the number of variables in
$Q$.

Consequently, researchers  have naturally sought  relaxed upper bounds to the combinatorial
bound that are (efficiently) computable, and that are as tight as
possible~\cite{
AGM,
GLVV,
panda,
DBLP:conf/pods/000118,
DBLP:conf/icdt/GogaczT17,
DBLP:conf/icdt/000122,
DBLP:journals/corr/abs-2502-05912}.
Most of these bounds are characterized by two collections:
the collection $\calF$ of set-functions $h : 2^{[n]} \to \R_+$ we are optimizing over, and
the collection $\dc$ of (degree-) constraints the functions have to satisfy.
Here, $n$ is the number of variables in the query.
Abusing notations, we also use $\dc$ to denote the constraints over set-functions:
$h(Y)-h(X)\leq c$, for every $(X,Y,c) \in \dc$.
Given $\dc$ and $\calF$, define
\begin{align}
    \dc[\calF] := \sup\{ h([n]) \ | \ h \in \calF \cap \dc \}.
    \label{eqn:dcF}
\end{align}
Different classes of $\calF$ and $\dc$ give rise to classes of
bounds with varying degrees of tightness (how close to the combinatorial bound) and
computational complexities~\cite{panda}.
This paper focuses on the {\em polymatroid
bound}~\cite{panda,DBLP:journals/corr/abs-2402-02001}, obtained by setting
$\calF = \Gamma_n$, the collection of all $n$-dimensional polymatroid functions.
The upper bound $\log_2 \sup_{\bm I\models \dc} |Q(\bm I)| \leq \dc[\Gamma_n]$
is the \emph{tightest} known upper bound
that is plausibly computable in polynomial-time~\cite{DBLP:conf/icdt/000122}.
(See Section~\ref{sec:background} for more.)

More precisely, the {\em polymatroid bound} $\dc[\Gamma_n]$
for a collection $\dc = \{(X_i, Y_i, c_i) \mid i \in [k] \}$ of $k$ degree constraints over
variables $X_i \subseteq Y_i \subseteq[n]$, can be computed by solving:
\begin{equation}
\begin{array}{rrlll}
\displaystyle \max & \multicolumn{3}{l}{h([n]) } \\
\textrm{s.t.} &  h(Y_i) - h(X_i) &\le&  c_i  &\forall i \in [k] \\
&\displaystyle h  &\in & \Gamma_n
\end{array}
\label{eqn:polymatroid-bound}
\end{equation}
The optimization problem~\eqref{eqn:polymatroid-bound} is a linear program (LP, see
Sec.~\ref{sec:background}). with an {\em exponential} number of variables and constraints.
In the query processing pipeline, for every query the optimizer has to issue {\em many}
cardinality estimates when searching for an optimal query plan. Hence, solving the
LP~\eqref{eqn:polymatroid-bound} in polynomial time is crucial.

While it is an {\em open} question whether the polymatroid bound can be computed in
poly-time for general degree constraints, it is known to be
computable in poly-time if the input degree constraints fall into one of the following
cases:
(a) all degree constraints are cardinality
  constraints~\cite{AGM,DBLP:journals/corr/abs-2402-02001}; in this case the polymatroid
  bound is the same as the AGM bound;
(b) the {\em constraint dependency graph} is {\em acyclic}~\cite{DBLP:conf/pods/000118};
  this is a generalization of the all cardinality constraints case;\footnote{
  The constraint dependency graph is the graph whose vertices are the variables $[n]$ and
  there is a directed edge $(u, v)$ if and only if there exists a degree constraint
  $(X_i, Y_i, c_i)$ where $u \in X_i$ and $v \in Y_i - X_i$.}
or (c) the constraints include cardinality constraints and {\em simple} functional
  dependencies (FD)~\cite{GLVV};
  Simple FDs are degree constraints of the form $(X, Y, c)$ with $c=0$ and $|X|=1$

Our focus in this paper will be on another class of commonly occurring types of degree
constraints called {\em simple} degree constraints. A degree constraint $(X, Y, c)$ is {\em
simple} if $|X| \leq 1$. This class strictly generalizes cardinality constraints and simple
functional dependencies.
Simple degree constraints also occur in the context of query
containment under bag semantics, where they are the key ingredient for the rare
special case when the problem is known to be decidable~\cite{DBLP:conf/pods/KhamisK0S20}. In
particular, Lemma 3.13 from~\cite{DBLP:conf/pods/KhamisK0S20} showed that, for simple degree
constraints, there is always an optimal polymatroid function of a special type
called a {\em normal} function. The proof was a rather involved constructive
argument that showed how to iteratively modify any feasible polymatroid function to obtain a
feasible normal function without decreasing the objective value. While this observation
offered some structural insight, it certainly did not resolve the issue of whether the
polymatroid bound could be efficiently computed when the degree constraints are simple.

The polymatroid bound~\eqref{eqn:polymatroid-bound} is deeply connected to
conjunctive query evaluation thanks to the PANDA
framework~\cite{panda,DBLP:journals/corr/abs-2402-02001}.
From an optimal solution to the dual $D$ of the LP~\eqref{eqn:polymatroid-bound},
~\cite{panda} showed how to construct a specific
linear inequality that is valid for all polymatroids called a {\em Shannon-flow
inequality}.
The validity of the Shannon-flow inequality can
be proved by constructing a specific sequence of elemental Shannon inequalities.
This sequence is called a
{\em proof-sequence} for the Shannon-flow inequality. From a proof sequence of length
$\ell$, the PANDA algorithm can answer a conjunctive query in time $O(|\bm I|)+ (\log |\bm
I|)^{O(\ell)} U$, where $U = 2^{\dc[\Gamma_n]}$ is the upper bound on the output size given
by the polymatroid bound,
and $\bm I$ is the input database instance.

While PANDA is a {\em very general} framework
for algorithmically constructing completely non-trivial query plans that optimize for the
total number of tuples scanned, the major drawback in the runtime expression is
that as the length of the proof sequence appears in the exponent. Previous
results~\cite{panda} can only establish $\ell$ to be {\em exponential} in $n$. A shorter
proof sequence would drastically improve the running time of PANDA.

The problems outlined above set the context for our questions and contributions. First,
under simple degree constraints is it possible to efficiently compute the polymatroid bound?
Second, can a proof sequence of polynomially bounded length be constructed? And third, what
is a good relaxation of the polymatroid bound that is tractable?

\subsection{Our Contributions}
\label{subsec:contributions}

Our first main contribution, stated
in Theorem \ref{thm:main1},  is to  show that the polymatroid bound can be modeled by a
{\em polynomial-sized} LP for simple degree constraints, and thus is computable in
poly-time when the degree constraints are simple.

\begin{thm}
Let $\dc$ be a collection of $k$ simple degree constraints over $n$ variables. The
polymatroid bound $\dc[\Gamma_n]$ can be modeled by a  LP where the
number of variables  is $O(kn^2)$ and the number of constraints is $O(kn)$. Thus the
polymatroid bound is  computable in  time polynomial in $n$ and $k$.
\label{thm:main1}
\end{thm}

We prove Theorem \ref{thm:main1} in Section \ref{sec:polymatroidbound}, but let us give a
brief overview here. The proof outline is shown schematically in
Figure~\ref{fig:simple:dpd}. The LP formulation $P$ of the polymatroid bound $\dc[\Gamma_n]$
in~\eqref{eqn:polymatroid-bound} has exponentially many variables and exponentially many
constraints, and so does $P$'s dual $D$. The starting point is the observation that there
are only polynomially many constraints in $P$ where the constant right-hand side is non-zero
(namely the constraints that correspond to the degree constraints), and thus the objective
of $D$ only has polynomial size. By projecting the feasible region for the LP $D$ down onto
the region of the variables in the objective, we show how to obtain an LP formulation
$\Dsimple$ of the polymatroid bound that has polynomially many variables, but still
exponentially many constraints. The most natural interpretation of $\Dsimple$ is as a sort
of min-cost {\em hypergraph} cut problem. We  show that this hypergraph cut problem can be
reformulated as $n$  min-cost network flow problems in a standard graph.  This results in a
natural polynomially-sized LP formulation $\Dflow$.
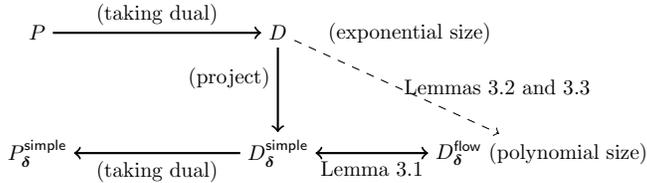
\begin{figure}[!ht]
\centering \begin{tikzpicture}[domain=0:20,
every node/.style={font=\large}, scale=0.8, every node/.style={scale=0.8}]

\node[] at (0, 2) (P) {$P$};
\node[] at (4, 2) (D) {$D$};
\node[right=1em of D] {(exponential size)};
\node[] at (4, 0) (Dd) {$\Dsimple$};
\node[] at (0, 0) (Pd) {$P_{\bm \delta}^{\sf simple}$};
\draw[thick,->] (P) -- (D) node [midway,above] {(taking dual)};
\draw[thick,->] (D) -- (Dd) node [midway,left,yshift=0.2cm] (project) {(project)};
\draw[thick,->] (Dd) -- (Pd) node [midway,below] {(taking dual)};
\node[right=1.5cm of Dd] (Dflow) {$\Dflow$ (polynomial size)};
\draw[<->, thick] (Dd) -- (Dflow) node [midway,below] {Lemma~\ref{lem:simplecut}};
\draw[dashed,->] (D) -- (Dflow) node [midway,right] (project2) {Lemmas~\ref{lem:simpleliftA1} and~\ref{lem:simpleliftA2}};

\end{tikzpicture}
\caption{Outline of the proof of Theorem \ref{thm:main1}.}
\label{fig:simple:dpd}
\end{figure}

As an ancillary result, we observe in Section \ref{sec:polymatroidbound} that the dual of
$\Dsimple$ is a natural LP formulation of the polymatroid bound when the
polymatroidal functions are restricted to normal functions. This gives an alternate proof,
that  uses LP duality, of the fact that for simple degree constraints there
is an optimal polymatroid function that is normal. We believe that this dual-project-dual
proof is simpler and more informative than the inductive proof
in~\cite{DBLP:conf/pods/KhamisK0S20}.

Our second main contribution, stated in Theorem \ref{thm:main2}, is to show that there is a
polynomial length proof sequence for simple instances. Beyond improving PANDA's runtime,
this theorem is a fundamental result about the geometry of submodular functions.

\begin{thm}
  Let $\dc$ be a collection of $k$ simple degree constraints over $n$ variables. There is a
  polynomial-time algorithm that computes  a   proof sequence for the  polymatroid bound of
  length $O(  k^2 n^2 + kn^3)$.
  \label{thm:main2}
\end{thm}

The  proof of Theorem \ref{thm:main2} in Section \ref{sec:ps:simple} constructively shows
how to create a proof sequence from a feasible solution to our LP $\Dflow$,
where the length of the proof sequence is linear in the cardinality of the support of the
feasible solution.  Intuitively, this proof sequence is exponentially shorter than the one
that would be produced using the method in \cite{panda} (even under simple degree constraints)
because it utilizes the special
properties of the types of feasible solutions for the LP $D$ that are derivable
from feasible solutions to our LP $\Dflow$. And because the proof sequence is
exponentially shorter, this exponentially decreases the run time when using the PANDA
algorithm.

We next prove that allowing other types of degree constraints that are ``just beyond
simple'' in some sense results in problems that are as hard as general constraints.
The proof of the following theorem is in Section \ref{sec:lower-bound}.

\begin{thm}
  Even if the set of input degree constraints $\dc$ is restricted to either any one of
  the following cases, computing $\dc[\Gamma_n]$ is still
  as hard as computing the polymatroid bound for general degree constraints:
  \begin{itemize}
    \item[(a)] If $\dc$ is a union of a set of simple degree constraints and a set of
      acyclic degree constraints.
      \item[(b)] If $\dc$ is a union of a set of simple degree constraints and a set of
      FD constraints.
    \item[(c)] If $\dc$ contains only degree constraints $(X, Y, c)$ with $|X|\leq 2$ and
    $|Y|\leq 3$.
  \end{itemize}
  \label{thm:not:much:better}
\end{thm}

As for simple degree constraints there is always an optimal polymatroid function that is
normal, one might plausibly conjecture that an explanation for the polymatroid bound being
efficiently computable for simple degree constraints is that the optimal normal function is
efficiently computable for general degree constraints. We  show that this conjecture is
highly unlikely to be true by showing that it is NP-hard to compute the optimal normal
function on general  instances. This is stated in Theorem \ref{thm:normalhard}, which is
proved in Section  \ref{sec:hardness1}.

\begin{thm}
  \label{thm:normalhard}
      Given a collection $\DC$ of degree constraints, it is NP-hard to compute the polymatroid bound
      $\DC[\Gamma_n]$ when the functions are additionally restricted to be normal.
\end{thm}

Theorem \ref{thm:not:much:better} implies that, perhaps, we have reached (or are at least
nearing) the limits of natural classes of degree constraints for which computing the
polymatroid bound is easier than computing the polymatroid bound for general instances. We
next seek a tractable relaxation of the polymatroid bound. Section \ref{sec:flow:bound}
introduces a new upper bound called the {\em flow bound}, denoted by $\flowbound(\dc, \pi)$,
which is parameterized by a permutation $\pi$ of the variables. The flow bound is strictly
better than the previously known tractable relaxation called the {\em chain bound}
$\dc_\pi[\Gamma_n]$ (See~\cite{DBLP:conf/pods/000118} and Sec~\ref{sec:background}), thus it
can be used as a tighter yet poly-time computable cardinality estimator.

\begin{thm}
  The flow-bound satisfies the following properties:
  \begin{itemize}
    \item[(a)] For all collections $\dc$ of degree constraints, and any given
    permutation $\pi$ of the variables,
    the flow bound can be computed in polynomial time, and is tighter than the chain bound, that is:
      \begin{align}
        \dc[\Gamma_n] \le \flowbound(\dc, \pi) \le \dc_\pi[\Gamma_n]
      \end{align}
      \item[(b)] If $\dc$ is either simple or acyclic, then in polynomial time in $n$ and $|\dc|$,
      we can compute a permutation $\pi$ such that $\flowbound(\dc, \pi) = \dc[\Gamma_n]$.
      \item[(c)] There are classes of instances $\dc$ where there exists permutations $\pi$
      for which the ratio-gap $\frac{\dc_\pi[\Gamma_n]}{\flowbound(\dc, \pi)}$ is unbounded
      above.
  \end{itemize}
  \label{thm:flow:bound}
\end{thm}

\section{Background}
\label{sec:background}

This section presents the minimal background required to understand the results in this
paper.

\subsection{Classes of set functions}

Given a function $h : 2^{[n]} \to \R_+$, and $X \subseteq Y \subseteq [n]$, define $$h(Y|X) =
h(Y) - h(X).$$ The function $h$ is {\em monotone} if $h(Y|X) \geq 0$ for all $X\subseteq Y$; it
is {\em submodular} if $h(I | I \cap J) \geq h(I\cup J | J)$, $\forall I,J\subseteq [n]$.
It is a {\em polymatroid} if it is non-negative, monotone, submodular, with $h(\emptyset)=0$.

For every $W \subsetneq [n]$, a {\em step function} $s_W : 2^{[n]} \to \R_+$
is defined by
\begin{align}
    s_W(X) = \begin{cases}
        0 & \text{if } X \subseteq W \\
        1 & \text{otherwise}
    \end{cases}
\end{align}
A function $h$ is {\em normal}  (also called a {\em weighted coverage
function}) if it is a non-negative linear combination of step functions.
A function $h$ is {\em modular} if there is a non-negative value
$w_i$ for each variable $i \in [n]$ such that for all $X \subseteq [n]$
it is the case that $h(X) = \sum_{i \in X} w_i$.

Let $\Gamma_n$, $M_n$, $\Nor_n$ denote the set of all polymatroid, modular, normal functions
over $[n]$, respectively. It is known~\cite{panda} that $M_n \subseteq N_n \subseteq
\Gamma_n$. Note that all three sets are polyhedral. In other words, when we view each
function as a vector over $2^{[n]}$, the set of vectors forms a convex polyhedron in that
vector space, defined by a finite number of hyperplanes.
To optimize linear objectives over these sets is to solve linear programs.

\subsection{Shannon-flow inequalities and proof sequences}

Let $\calP \subseteq 2^{[n]} \times 2^{[n]}$ denote the set of all pairs $(X,Y)$
such that $\emptyset \subseteq X \subsetneq Y \subseteq [n]$.
Let $\bm \delta = (\delta_{Y|X})_{(X,Y) \in \calP}$ be a vector of non-negative reals.
The inequality
\begin{align}
    h([n] | \emptyset) \leq \sum_{(X, Y) \in \calP} \delta_{Y | X} \cdot h(Y | X)
    \label{eqn:sfi}
\end{align}
is called a {\em Shannon-flow inequality}~\cite{panda} if it is satisfied by all polymatroids
$h \in \Gamma_n$.

One way to prove that~\eqref{eqn:sfi} holds for all polymatroids is to turn the RHS into
the LHS by repeatedly applying one of the following replacements:
\begin{itemize}
    \item  {\em Decomposition:} $h(Y|\emptyset) \to h(X|\emptyset) + h(Y|X)$, for $X \subsetneq Y$.
    \item {\em Composition:} $h(X|\emptyset) + h(Y|X) \to h(Y|\emptyset)$, for $X \subsetneq Y$.
    \item {\em Monotonicity:} $h(Y|X) \to 0$, for $X \subseteq Y$.
    \item  {\em Submodularity:} $h(I|I\cap J) \to h(I\cup J|J)$, for $I \perp J$,
    which means $I \not\subseteq J$ and $J \not\subseteq I$.
\end{itemize}
Note that, we can replace $\to$ by $\geq$ in all four cases above to obtain valid
Shannon inequalities that are satisfied by all polymatroids.
Each replacement step is called a
``{\em proof step},'' which can be multiplied with a {\em non-negative} weight $w$.
For example, $w \cdot h(I |I \cap J) \geq w \cdot h(I \cup J | J)$ is a valid inequality.
We prepend $w$ to the name of the operation to denote the weight being used, so we will refer to the prior
inequality as an $w$-submodularity step.
Note also that, the terms $h(Y|X)$ are manipulated as symbolic variables.

A {\em proof sequence} for the inequality~\eqref{eqn:sfi} is a sequence of proof steps, where
\begin{itemize}
    \item Every step is one of the above proof steps, accompanied by a non-negative weight $w$
    \item After every step is applied, the coefficient of every term $h(Y|X)$ remains
    non-negative.
    \item The sequence starts with the RHS of~\eqref{eqn:sfi} and ends with the LHS.
\end{itemize}
Every proof step thus transforms a collection of non-negatively weighted (``conditional
polymatroid'') terms into another collection of non-negatively weighted terms. One of the
main results in the paper~\cite{panda} stats that,
the inequality~\eqref{eqn:sfi} is a Shannon-flow inequality {\em if and only if} there exists
a proof sequence for it.

\subsection{Comparison of polymatroid bound to other bounds}
\label{subsec:boundscomparisons}

Recall the bound $\dc[\calF]$ defined in~\eqref{eqn:dcF}. By parameterizing this bound with
combinations of $\calF$ and $\dc$, we obtain a hierarchy of bounds, which we briefly
summarize here. By setting $\calF = \Gamma_n$, $\Nor_n$, $M_n$, we obtain the {\em
polymatroid bound} $\dc[\Gamma_n]$, the {\em normal bound} $\dc[\Nor_n]$, and the {\em
modular bound} $\dc[M_n]$, respectively. For a permutation $\pi$ of $[n]$, let $\dc_\pi$
denote the collection of degree constraints obtained  by modifying each $(X_i,Y_i,c_i) \in
\dc$ by retaining the variables in $Y_i$ that either are in $X_i$ or are listed after all
variables in $X_i$ in the $\pi$-order. Note that $\dc_\pi$ is a relaxation of $\dc$ in the
sense that if a database instance satisfies $\dc$, then it  also satisfies $\dc_\pi$. The
following inequalities are known \cite{DBLP:conf/pods/000118,suciu2023applications}:
\begin{align}
    \dc[M_n] &\leq \dc[\Nor_n]
    \leq \dc[\Gamma_n] \leq \min_\pi \dc_\pi[\Gamma_n] \\
   \log \max_{\bm I \models \dc} |Q(\bm I)| &\leq \dc[\Gamma_n] \\
   \dc[M_n] &\leq n \cdot \log \max_{\bm I \models \dc} |Q(\bm I)| \\
   \dc[\Nor_n] &\leq 2^n \cdot \log \max_{\bm I \models \dc} |Q(\bm I)|
\end{align}
Fix a permutation $\pi$ (a ``chain''), then the {\em chain-bound}
$\dc_\pi[\Gamma_n]$ is computable in poly-time~\cite{DBLP:conf/pods/000118}.
If $\dc$ is acyclic, then we can compute a permutation $\pi$ (in poly-time)
such that $\dc_\pi[\Gamma_n] = \dc[M_n]$.

\subsection{Additional Background}
\label{subsetc:background}
We illustrate how the various bounds compare using the following triangle query as a running example:

\begin{align}
   Q_{\Delta}(a,b,c) \leftarrow R(a,b)\wedge S(b,c)\wedge T(a,c).
\end{align}
Suppose we know the cardinalities of $R$, $S$, and $T$. We can then derive a series of output
size bound for this query as follows.
Let $\bm 1_{R(a,b)}, \bm 1_{S(b,c)}, \bm 1_{T(a,c)}$ denote the indicator variables for
whether $(a,b)\in R, (b,c) \in S$, and $(a,c)\in T$, respectively. Then, the following
optimization problem captures the {\em exact} size bound for this problem:
\begin{align*}
    \max & \qquad \sum_{a,b,c} \bm 1_{R(a,b)}\bm 1_{S(b,c)}\bm 1_{T(a,c)}\\
    \text{s.t.} & \qquad \sum_{a,b}\bm 1_{R(a,b)} \leq |R| \\
                &\qquad    \sum_{a,b}\bm 1_{S(b,c)} \leq |S| \\
                &\qquad  \sum_{a,c}\bm 1_{T(a,c)} \leq |T|.
\end{align*}
While the bound is exact, it is certainly impractical and computing it this way may even
be worse than computing the query itself.

\paragraph*{The entropic bound}
Chung et al.~\cite{MR859293} gave the following {\em entropy argument}.
Consider the input to the query that gives the maximum output size $|Q_\Delta|$.
Pick a triple $(a,b,c)$ from this worst-case $Q_\Delta$ uniformly at random, and let $H$
denote the entropy of this uniform distribution. Then, it is easy to see that
\begin{align*}
    H(a,b,c) &= \log |Q_\Delta| = \log \text{combinatorial-bound} \\
    H(a,b) &\leq \log |R|, \ H(b,c) \leq \log |S|, \ H(a,c) \leq \log |T|.
\end{align*}
Thus, let $\Gamma^*_3$ denote the set of all entropic functions on finite and discrete
3D-distributions, the following is an upper bound for $|Q_\Delta|$:
\[
\max_{h \in \Gamma^*_3} \{ h(a,b,c) \suchthat h(a,b)\leq \log|R|, h(b,c)\leq \log |S|,
        h(a,c)\leq \log|T|\}
\]

\paragraph*{The polymatroid bound}
Thanks to Shannon~\cite{Yeung:2008:ITN:1457455}, we have known  for a long time that every
entropic function is a polymatroidal function (or {\em polymatroid} for short)
These are functions $h : 2^{\{a,b,c\}} \to \R_+$ that are
monotone and submodular (see Sec~\ref{sec:background}).
In particular, $\Gamma^*_3 \subseteq \Gamma_3$, which is the
set of all polymatroids over $3$ variables. The entropic bound can be relaxed into:
\[
\max_{h \in \Gamma_3} \{ h(a,b,c) \suchthat h(a,b)\leq \log|R|, h(b,c)\leq \log |S|,
        h(a,c)\leq \log|T|\}
\]
The main advantage of this so-called {\em polymatroid bound} is that it is a linear program,
because $\Gamma_3$ is a polyhedron and all constraints are linear.
While the entropic bound is not known to be computable at all (for more than $3$
dimensions), the polymatroid bound is computable by solving a linear program.
However, there are exponentially (in $n$) many facets in the polyhedron
$\Gamma_n$; thus, while computable the bound cannot be used as is
in practice.

\paragraph*{The modular bound}
An idea from Lovasz~\cite{MR717403} called ``modularization'' can be used to reduce the
dimensionality of the above polymatroid bound. Given any feasible solution $h$ to the
polymatroid LP, define a {\em modular} function $g :
2^{\{a,b,c\}}\to \R_+$ by :
\[ g(a) := h(a), \ g(b) := h(a,b)-h(a), \ g(c) := h(a,b,c)-h(a,b). \]
Then, extend modularly: $g(X) := \sum_{x \in X} g(x)$ for any $X \subseteq \{a,b,c\}$.
It is easy to see that $g$ satisfies all
constraints that $h$ satisfies, because, due to $h$'s submodularity,
\[ g(a,b)=h(a,b), \ g(b,c) \leq h(b,c), \ g(a,c) \leq h(a,c), \]
and due to telescoping $h(a,b,c) = g(a,b,c)$.
Let $\Mod_3$ denote the set of all modular functions on three variables, then, $\Mod_3
\subseteq \Gamma_3$. Thus, the polymatroid bound is {\em equal} to the following ``modular
bound''
\begin{align}
    \max &\qquad g(a)+g(b)+g(c) \label{eqn:modular:bound}\\
    \text{s.t.} &\qquad g(a)+g(b)\leq \log|R|, \ \ g(b)+g(c) \leq \log |S|,\nonumber \\
                &\qquad g(a)+g(c) \leq \log|T|, \ g(a),g(b),g(c)\geq 0. \nonumber
\end{align}
The beauty of the modular bound is that the LP now has polynomial size in $n$ and $|\dc|$,
and thus can be computed efficiently.

\paragraph*{The AGM bound}
The acute reader might have noticed that the LP~\eqref{eqn:modular:bound} is the
vertex-packing LP of the triangle graph. The {\em dual} of that LP, the {\em fractional edge
cover} LP, is known as the AGM-bound~\cite{AGM} for the triangle query.

\paragraph*{Tightness}
In the quest to obtain an efficiently computable bound, we have relaxed the combinatorial
bound to the entropic bound, polymatroid bound, which is equal to the modular and AGM
bounds. How much did we lose in this gradual relaxation process? The answer comes from
Atserias, Grohe, and Marx~\cite{AGM}, where the modular bound~\eqref{eqn:modular:bound}
plays the key role. Let $g$ denote an optimal solution to~\eqref{eqn:modular:bound}, and
define the database instance:
$R = [\lfloor 2^{g(a)}\rfloor] \times [\lfloor 2^{g(b)}\rfloor]$,
$S = [\lfloor 2^{g(b)}\rfloor] \times [\lfloor 2^{g(c)}\rfloor]$,
$T = [\lfloor 2^{g(a)}\rfloor] \times [\lfloor 2^{g(c)}\rfloor]$,
Then
\begin{align*}
    \text{combinatorial-bound} \geq
    |R \Join S \Join T|
    \geq \frac 1 8 2^{g(a)+g(b)+g(c)} \geq \frac 1 8 \text{combinatorial-bound},
\end{align*}
in particular, all bounds described thus far
are within $1/8$ of the combinatorial-bound for the triangle query.

The above line of derivations works almost verbatim for an arbitrary full conjunctive query
where all input relations are EDBs and the only thing we know about them are their
cardinalities. That is one way to derive the AGM-bound~\cite{AGM}.

\section{Computing the polymatroid bound in polynomial time}
\label{sec:polymatroidbound}
\subsection{Review of the linear programming formulation}

This subsection reviews the relevant progress made in \cite{panda}. The LP
formulation for computing the polymatroid bound $\mathrm{DC}[\Gamma_n]$ on a
collection $\dc = \{(X_i, Y_i, c_i) \mid i \in [k] \}$ of degree constraints was shown
in equation~\eqref{eqn:polymatroid-bound}.
We now write down this LP more explicitly, by listing all the constraints
defining the polymatroids.
To make the formulation more symmetrical,
in the following, we do not restrict $h(\emptyset)=0$; instead of maximizing $h([n])$,
we maximize the shifted quantity $h([n]) - h(\emptyset)$.
The function $h'(X) = h(X) - h(\emptyset)$ is a polymatroid.

There is a variable $h(X)$ for each subset $X$ of  $[n]$. The LP, denoted by $P$, is:
 \begin{equation}
\begin{array}{rrlll}
\displaystyle P:& \multicolumn{3}{l}{\max  \quad h([n]) - h(\emptyset)} \\
\textrm{s.t.} &  h(Y\cup X) - h(X) - h(Y) +  h( Y \cap X) & \leq & 0 & \forall X \forall Y,  X \perp Y  \\
&h(Y) - h( X) & \geq & 0 &\forall X \forall Y,  X \subsetneq Y \\
&\displaystyle h(Y_i) - h(X_i) &\le& c_i  &\forall i \in [k]
\end{array}
\label{eqn:LP-P}
\end{equation}
where  $X \incomp Y$ means $X \not\subseteq Y$ and $Y \not\subseteq X$. The first collection
of constraints enforce that the function $h$ is submodular, the second enforces that the
function $h$ is monotone, and the third enforces the degree constraints. We will adopt the
convention that all variables in our LPs are constrained to be {\em non-negative}
unless explicitly mentioned otherwise. Note that the linear program $P$ has both
exponentially many variables and exponentially many constraints.

To formulate the dual LP $D$ of $P$, we associate dual variables
$\sigma_{X,Y}$, dual variables $\mu_{X,Y}$ and dual variables $\delta_i$ with the three
types of constraints in $P$ (in that order). The dual $D$ is then:

\begin{equation}
\begin{array}{rrclcl}
\displaystyle D: \quad \quad \min & \multicolumn{3}{l}{\sum_{i \in [k]} c_i \cdot \delta_i} \\
\textrm{s.t.} & \flow([n])  & \geq & 1  \\
&\flow(\emptyset)  & \geq & -1\\
&\flow(Z)              & \geq& 0, \qquad \forall Z  \ne \emptyset, [n]
\end{array}
\label{eqn:LP-D}
\end{equation}
where $\flow(Z)$ is
defined by:
\begin{multline*}
   \flow(Z) :=
     \sum_{i: Z = Y_i} \delta_i -
   \sum_{i: Z=X_i }\delta_i +
   \sum_{\substack{I\incomp J\\I\cap J = Z}}\sigma_{I,J}
   +
   \sum_{\substack{I'\incomp J'\\I'\cup J' = Z}}\sigma_{I',J'}
   - \sum_{J: J\incomp Z}\sigma_{Z,J}-
   \sum_{X: X\subsetneq Z}\mu_{X,Z}+ \sum_{Y: Z\subsetneq Y}\mu_{Z,Y}
\end{multline*}
We use $\bm \delta, \bm \sigma, \bm \mu$ to denote the vectors of $\delta_i$,
$\sigma_{X,Y}$, and $\mu_{X,Y}$ variables.

The dual $D$ can be interpreted as encoding a  min-cost flow problem on a hypergraph
$\calL$. The vertices of $\mathcal L$ are the  subsets of $[n]$. Each variable $\mu_{X,Y}$
in $D$ represents the flow on a directed edge in $\calL$ from $Y$ to $X$. Each variable
$\delta_i$ in $D$ represents both the flow and capacity on the directed edge from $X_i$ to
$Y_i$ in $\calL$. The cost to buy this capacity is $\delta_i \cdot c_i$. Each variable
$\sigma_{X,Y}$ in $D$ represents the flow leaving each of $X$ and $Y$, and entering each of
$X \cap Y$ and $X \cup Y$ through the hyperedge containing four vertices.  The $\sigma$
variables involve flow between four nodes, which is why this is a hypergraph flow problem,
not a graph flow problem. Flow can not be created at any vertex other than the empty set,
and is conserved at all vertices other than the empty set and $[n]$.  That is, there is flow
conservation at each vertex like a standard flow problem. The objective is to  minimize the
cost of the bought capacity subject to the constraint that this capacity can support a unit
of flow from the source $s=\emptyset$ to the sink $t=[n]$.

\begin{ex}[Running Example Instance] It is challenging to construct a small example that
illustrates all the interesting aspects of our algorithm design, but the following
collection of degree constraints is sufficient to illustrate many interesting aspects:
$$
(1)\;  h(ab) \le  1;\quad
(2)\; h(bc) \le 2; \quad
(3)\;  h(ac) \le  1; \quad
(4)\;  h(ad) - h(a) \le  1
$$
These degree constraints are indexed (1) through (4) as indicated and $k=4$. The optimal
solution for $D$ to buy the following capacities, and route flow as follows:

\begin{figure}[h]
\begin{tikzpicture}[scale=1.5, every node/.style={scale=1.0}]
  \node (abcd) at (1,3) { $\{a, b, c, d\}$ };
  \node (abc) at (-1,2) { $\{a, b, c\}$ };
  \node (ab) at (-2,1) { $\{a, b\}$ };
  \node (ac) at (0,1) { $\{a, c\}$ };
  \node (ad) at (2,1) { $\{a, d\}$ };
  \node (a) at (-1,0) { $\{a\}$ };
  \node (empty) at (-1,-1) { $\emptyset$ };
  \node[text=red] (sig1) at (-2,1.5) { $\sigma_{\{a,b\}, \{a,c\}}$ };
  \node[text=blue] (sig2) at (1,1.5) { $\sigma_{\{a,b,c\}, \{a,d\}}$ };

  \draw[thick,red] (abc) -- (ab);
  \draw[thick,red] (abc) -- (ac);
  \draw[thick,red] (ab) -- (a);
  \draw[thick,red] (ac) -- (a);
  \draw[thick,blue] (abcd) -- (ad);
  \draw[thick,blue] (ad) -- (a);
  \draw[thick,blue] (a) -- (abc);
  \draw[thick,blue] (abcd) -- (abc);

  \draw[>=stealth,<-, orange] (empty) -- (a) node[midway, left] {$\mu_{\emptyset,\{a\}}$};
  \draw[>=stealth,->, bend left=45] (empty) to node[left] {$\delta_{\emptyset,\{a,b\}}$} (ab);
  \draw[>=stealth,->, bend right=45] (empty) to node[right] {$\delta_{\emptyset,\{a,c\}}$} (ac);
  \draw[>=stealth,->, bend right=25] (a) to node[right] {$\delta_{\{a\},\{a,d\}}$} (ad);
\end{tikzpicture}
\label{fig:flow}
\caption{The feasible flow in the optimal solution for $D$,  where all depicted variables
are set to $1$.}
\end{figure}
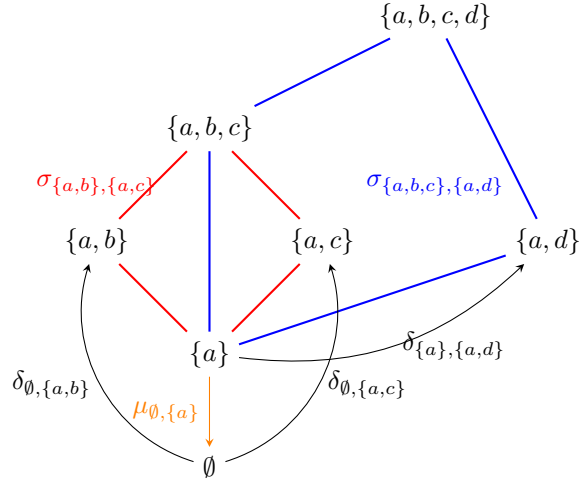

Set $\delta_1, \delta_3,\delta_4 $ equal to $1$ and $\delta_2$ to $0$.  The objective will
be $3$.  This can support the flow as follows.  A flow of $1$ is sent from $\emptyset$ to
$\{a,b\}$ and from $\emptyset$ to $\{a,c\}$ on the corresponding degree constraints. Then a
flow of $1$ is sent from $\{a,b\}$ and $\{a,c\}$ to both $\{a,b,c\}$ and $\{a\}$ using
$\sigma_{\{a,b\},\{a,c\}}$. The flow of $1$ at $\{a\}$ is sent to $\{a,d\}$ on the degree
constraint $\delta_4$.  A flow of $1$ is then sent to both $\{a,b,c,d\}$ and $\{a,\}$ from
$\{a,b,c\}$ and $\{a,d\}$ using $\sigma_{\{a,b,c\}, \{a,d\}}$. Thus, a flow of $1$ reaches
the set of all elements. Finally the leftover flow of $1$ at $\{a\}$ is returned to
$\emptyset$ on $\mu_{\emptyset,\{a\}}$.
\end{ex}

\subsection{Polynomial sized linear programming formulation (proof of Theorem~\ref{thm:main1})}
\label{subsec:simple}

This is where we leave the foundation laid by \cite{panda} and begin covering new ground.
We shall show that the LP $D$ can be replaced by projecting its feasible region
down to the space of the $\delta_i$ variables, resulting in the following LP:
\begin{align}
    \Dsimple: && \min    & \qquad \sum_{i \in [k]} c_i \cdot \delta_i&& \label{eqn:DS} \\
         && \mbox{s.t.}  & \qquad \sum_{i \in [k]: \ X_i \subseteq V,  Y_i \not\subseteq V}
                        \delta_i \geq 1  & \forall V \subsetneq [n]   \nonumber
\end{align}
$\Dsimple$ can be interpreted as a min-cost cut problem in the hypergraph network $\mathcal
L$, where the goal is to cheaply buy sufficient capacity on the $\delta$ edges  so the
$\delta$ edges crossing various cuts, determined by the degree constraints, have in
aggregate at least unit capacity. Note that $\Dsimple$ has only linearly many variables (one
for each degree constraint), but exponentially many constraints.

We next present a linear program $\Dflow$ that shall be shown to be equivalent to $D$ and
$\Dsimple$, but with only polynomially many variables and only polynomially many
constraints. The intuition behind $\Dflow$ is that it has a natural interpretation as a
min-cost flow problem in a (directed graph) network $G=(V, E)$ with a single source $s=\emptyset$ and
one sink for each variable. The vertices in $V$ consist of the empty set $\emptyset$, the
singleton sets $\{i\}$ for $i \in [n]$, and the sets $Y_i$, $i \in [k]$. The edges in $E$
are all the degree constraint edges $(X_i, Y_i)$  from $\mathcal L$, and all the $\mu_{X,Y}$
edges from $\mathcal L$ where $X$ and $Y$ are vertices in $V$. The cost of the degree
constraint edges $(X_i, Y_i)$ remains $\delta_i$, and the costs of the $\mu_{X,Y}$ edges
remain 0. So $G$ is a subgraph of the hypergraph $\mathcal L$. See Figure ~\ref{fig:auxG}
for an illustration of $G$ for our running example. Now the  problem is to spend as little
as possible to buy enough capacity so that for all sinks/vertices $t \in [n]$ it is the case
that there is sufficient capacity to route a unit of flow from the source $s=\emptyset$ to
the sink $t$.

\begin{figure}[h]
    \centering
    \begin{minipage}{0.45\textwidth}
        \centering
        \begin{tikzpicture}[scale=1.0, every node/.style={scale=1.0}]
            \node (empty) at (0,0) {$\emptyset$};
            \node (a) at (0,1) {$\{a\}$};
            \node (b) at (0,2) {$\{b\}$};
            \node (c) at (0,3) {$\{c\}$};
            \node (d) at (0,4) {$\{d\}$};
            \node (ab) at (4,0.5) {$\{a,b\}$};
            \node (bc) at (4,1.5) {$\{b,c\}$};
            \node (ac) at (4,2.5) {$\{a,c\}$};
            \node (ad) at (4,3.5) {$\{a,d\}$};
            \draw[>=stealth,->,thick,blue] (empty) to node[left] {$1$} (ab);
            \draw[>=stealth,->,thick,blue] (empty) to node[right] {$2$} (bc);
            \draw[>=stealth,->,thick,blue] (empty) to node[left] {$1$} (ac);
            \draw[>=stealth,->,thick,DarkGreen, bend left=10] (ab) to (a);
            \draw[>=stealth,->,thick,DarkGreen] (ab) -- (b);
            \draw[>=stealth,->,thick,blue, bend left=20] (a) to node[left] {$1$} (ad);
            \draw[>=stealth,->,thick,DarkGreen] (ad) -- (d);
            \draw[>=stealth,->,thick,DarkGreen] (ad) -- (a);
            \draw[>=stealth,->,thick,DarkGreen] (ac) -- (a);
            \draw[>=stealth,->,thick,DarkGreen] (ac) -- (c);
            \draw[>=stealth,->,thick,DarkGreen] (bc) -- (b);
            \draw[>=stealth,->,thick,DarkGreen] (bc) -- (c);
        \end{tikzpicture}
        \caption{\label{fig:auxG} The auxiliary graph $G$ for the running example.
        The blue edges correspond to the degree constraints, with annotated costs.
        The green edges correspond to the $\mu$ variables, and cost $0$.}
    \end{minipage}\hfill
    \begin{minipage}{0.45\textwidth}
        \centering
        \begin{tikzpicture}[scale=1.0, every node/.style={scale=1.0}]
            \node (empty) at (0,0) {$\emptyset$};
            \node (a) at (0,1) {$\{a\}$};
            \node (b) at (0,2) {$\{b\}$};
            \node (c) at (0,3) {$\{c\}$};
            \node (d) at (0,4) {$\{d\}$};
            \node (ab) at (4,0.5) {$\{a,b\}$};
            \node (bc) at (4,1.5) {$\{b,c\}$};
            \node (ac) at (4,2.5) {$\{a,c\}$};
            \node (ad) at (4,3.5) {$\{a,d\}$};
            \draw[>=stealth,->,thick,blue] (empty) -- (ab);
            \draw[>=stealth,->,thick,DarkGreen, bend left=10] (ab) to (a);
            \draw[>=stealth,->,thick,blue, bend left=20] (a) to (ad);
            \draw[>=stealth,->,thick,DarkGreen] (ad) -- (d);
        \end{tikzpicture}

        \caption{\label{fig:auxGflow}A unit flow in the auxiliary graph $G$ from
        $\emptyset$ to $\{d\}$.  For this flow to be feasible, a unit capacity must be
        bought on these edges.}
    \end{minipage}
\end{figure}

This problem is naturally modeled by the following linear program $\Dflow$:
\begin{align}
    \Dflow: & &\min \sum_{i \in [k]} c_i \cdot \delta_i \label{eqn:Dflow} \\
          & \textrm{s.t.} &\qquad f_{i,t} &\leq \delta_i  && \forall i \in [k] & \forall t
          \in [n] \nonumber\\
          && \qquad \flow_t(t) &= 1  & &&  \forall t \in [n]\nonumber\\
          && \qquad \flow_t(\emptyset) &= -1  &&& \forall t \in [n]\nonumber\\
          && \qquad \flow_t(Z) &\geq 0  && \forall Z   \in  G \setminus \{\emptyset\}
    \setminus \{t\} & \forall t \in [n] \nonumber
\end{align}
where,
\begin{align*}
   \flow_t(Z) &:=
     \sum_{i: Z = Y_i} f_{i,t} -
   \sum_{i: Z=X_i }f_{i,t} +
   \sum_{X: X\subsetneq Z}\mu_{X,Z, t}+ \sum_{Y: Z\subsetneq Y}\mu_{Z,Y, t}
\end{align*}
The interpretation of $f_{i,t}$ is the flow routed from $X_i$ to $Y_i$ in $G$ for the flow
problem where the sink is $\{t\}$; $\mu_{Z,Y,t}$ is the flow routed from $Y$ to $Z$ in $G$.
So the first set of constraints say that the capacity
bounds are respected, the second and third set of constraints ensure that unit flow leaves
the source and arrives at the appropriate sink, and the last set of constraints ensure flow
conservation. Note that as the graph $G$ has  $O(k+ n)$ vertices and $O(kn)$ edges, the
LP $\Dflow$ has $O(kn^2)$ variables and $O(k  n)$ constraints.

\begin{ex}[Running example]
The optimal solution for $\Dflow$ on our running example is to buy unit capacity on the
$\mu$ edges for a cost of 0, buy unit capacity on degree constraint edges  $\emptyset \to ab$,
$\emptyset \to ac$, and $a \to ad$ for a cost of $1$ each, resulting
in a total cost of $3$. This supports a flow of $1$ to $a$ by routing a unit of flow on the path
$\emptyset \to ab \to a$, supports a flow of $1$ to $b$ by routing a unit of flow on the
path  $\emptyset \to ab \to b$, supports a flow of $1$ to $c$ by routing a unit of flow on
the path  $\emptyset \to ac \to c$, and supports a flow of 1 to $d$ by routing a unit of
flow on the path  $\emptyset \to ab \to a \to ad \to d$. This flow to $d$ is shown in
Figure \ref{fig:auxGflow}.
\end{ex}

We now formally show that the linear programs $D$, $\Dsimple$ and $\Dflow$ all have the same
objective function. Refer to Fig.~\ref{fig:simple:dpd} for the high-level plan. To show that
they are equivalent, it is sufficient to just consider the feasible regions for these LPs.
We prove equivalence in the following manner. Lemma \ref{lem:simplecut} shows feasible
regions of the LPs $\Dsimple$ and $\Dflow$ are identical. Then Lemma \ref{lem:simpleliftA1}
shows that the polyhedron defined by projecting the feasible region of the LP $D$ onto the
${\bm \delta}$-space is a subset of the feasible region defined by the LP $\Dflow$. Finally
Lemma \ref{lem:simpleliftA2} shows that polyhedron defined by projecting the feasible region
of the LP $D$ onto the ${\bm \delta}$-space is a superset of the feasible region defined by
the LP $\Dflow$.

\begin{lmm} \label{lem:simplecut}
The feasible regions of the linear programs $\Dsimple$ and $\Dflow$ are identical.
\end{lmm}

\begin{proof}
Assume  for some setting of the $\delta_i$ variables, that $\Dflow$ is infeasible. Then
there exists a $t \in [n]$ such that the max flow between the source $s=\emptyset$ and
$\{t\}$ is less than 1. Since the value of the maximum  $s$-$t$ flow is equal the value of
the minimum $s$-$t$ cut, there must be a subset $W$ of vertices in $G$ such that $s \in W$
and $t \notin W$, where the aggregate capacities leaving $W$ is less than one. By taking $V
:= \{ i \in [n] \ | \ \{i\} \in W \}$ we obtain a violated constraint for $\Dsimple$.

Conversely, assume that for some setting of the $\delta_i$ variables, that $\Dsimple$ is
infeasible. Then there is a set $V \subsetneq [n]$ such that $\sum_{i \in [k]: X_i \subseteq
V, Y_i \not\subseteq V} \delta_i < 1$. Fix an arbitrary $t \in [n] \setminus V$, and let $W
:= \{ s \} \cup \{ \{i\} \ | \ i \in V\}$. The $(s,t)$-cut $(W, V(G) \setminus W)$ has
capacity $\sum_{i \in [k]: X_i \subseteq V, Y_i \not\subseteq V} \delta_i$ which is strictly
less than $1$. This means $\bm \delta$ is not feasible for $\Dflow$, a contradiction.
\end{proof}

\begin{lmm} \label{lem:simpleliftA1}
The polyhedron defined by projecting the feasible region of the linear program
$D$ onto the ${\bm \delta}$-space is a subset of the feasible region defined by the
linear program $\Dflow$.
\end{lmm}
\begin{proof}
Let $(\bm \delta,\bm \sigma,\bm \mu)$ be a feasible solution to the linear program $D$,
where $\bm \delta = (\delta_i)_{i\in [k]}$. We show that $\bm \delta$ is feasible for the
linear program $\Dflow$. Assume to the contrary that $\bm \delta$ is {\em not} feasible for
$\Dflow$, then there exists a $t \in [n] $ such that there is a cut in the flow network $G$
with capacity $< 1$ that  separates  $s=\emptyset$ and $\{t\}$. In particular, let $V$  be
the union of all singleton sets that are on the same side of this cut as $s=\emptyset$;
then, $\sum_{i \in [k]: \ X_i \subseteq V, Y_i \not\subseteq V}\delta_i < 1$.

Now consider the flow hypernetwork $\mathcal L$ associated with the linear program $D$. Let
$W := \{s\} \cup \{ \{i\} \suchthat i \in V\}$ be a set of vertices in $\calL$. Then in the
hypergraph $\mathcal L$ we claim that the aggregate flow coming out of $W$ can be at most
$\sum_{i \in [k]: \ X_i \subseteq V, Y_i \not\subseteq V}\delta_i < 1$. Since the net flow
out of $W$ is equal to the total flow received by vertices not in $W$, which is $\sum_{Z
\notin W} \flow(Z) \geq 1$, we reach a contradiction.

To see that the claim holds, note that no $\mu_{X,Y}$ edge can cause flow to escape $W$, and
no $\sigma_{X,Y}$ hyperedge can contribute a positive flow to leave $W$: if $X \in W$ and
$Y\in W$ then $X \cup Y \in W$, and if either $X \notin W$ or $Y \notin W$ then
$\sigma_{X,Y}$ does not route any positive net flow out of $W$.
\end{proof}

\begin{lmm} \label{lem:simpleliftA2}
The polyhedron defined by projecting the feasible region of the linear program
$D$ onto the ${\bm \delta}$-space is a superset of the feasible region defined by the
linear program $\Dflow$.
\end{lmm}
\begin{proof}
To prove this lemma we constructively show how to extend a feasible solution $\bm \delta $
for $\Dflow$ to a feasible solution $(\bm \delta, \bm \sigma,\bm \mu)$ for $D$ by setting
$\bm \mu$ and $\bm \sigma$ variables. The extension is shown in
Algorithm~\ref{algo:delta:to:sigma:mu}.
\begin{algorithm}[th]
     \caption{Constructing a feasible solution $(\bm\delta,\bm\sigma,\bm\mu)$ to the LP $D$}
     \label{algo:delta:to:sigma:mu}

     \begin{algorithmic}[1]
          \State $\bm \mu \gets \bm 0$, $\bm \sigma \gets \bm 0$
          \For {$j = 1$ to $k$}
               \algorithmiccomment{For each degree constraint $(X_j,Y_j,c_j)$}
               \State $\mu_{X_j, Y_j} \gets \delta_{j}$
               \label{line:initialize:mu}
          \EndFor
          \For {$i = 0$ to $n-1$}
              \algorithmiccomment{Outer Loop}
              \State Let $\calP^{i+1}$ be the collection of simple flow paths
              routing $1$ unit of flow from $s=\emptyset$ to $t=\{i+1\}$
              \For {each path $P \in \calP^{i+1}$ with flow value $\epsilon$}
                    \algorithmiccomment{Forward Path Loop}
                    \For {each edge $(A,B) \in P$ from $\emptyset$ to $\{i+1\}$}
                         \If {$A \cup [i] \subsetneq B \cup [i]$}
                              \State Decrease $\mu_{A \cup [i], B \cup [i]}$ by $\epsilon$
                              \label{line:decrease:mu1}
                         \ElsIf {$B \cup [i] \subsetneq A \cup [i]$}
                              \State Increase $\mu_{B \cup [i], A \cup [i]}$ by $\epsilon$
                              \label{line:increase:mu1}
                         \EndIf
                    \EndFor
              \EndFor
              \For {each path $P \in \calP^{i+1}$ with flow value $\epsilon$}
                    \algorithmiccomment{Backward Path Loop}
                    \For {each edge $(A,B) \in P$ from $\{i+1\}$ back to $\emptyset$}
                         \If {$A \cup [i+1] \subsetneq B \cup [i+1]$}
                              \State Increase $\mu_{A \cup [i+1], B \cup [i+1]}$ by $\epsilon$
                              \label{line:increase:mu2}
                         \ElsIf {$B \cup [i+1] \subsetneq A \cup [i+1]$}
                              \State Decrease $\mu_{B \cup [i], A \cup [i]}$ by $\epsilon$
                              \If {$i+1 \in A$}
                                   \algorithmiccomment{This means $A\cup [i] = A\cup [i+1]$}
                                   \State Increase  $\mu_{B \cup [i],B \cup [i+1]}$  by $\epsilon$
                              \label{line:increase:mu3}
                              \Else
                                   \State Increase  $\sigma_{A \cup [i],B \cup [i+1]}$  by $\epsilon$
                              \label{line:increase:sigma1}
                              \EndIf
                         \EndIf
                    \EndFor
              \EndFor
              \For {each $j \in [k]$, where $X_j \cup [i+1] \subsetneq Y_j \cup [i+1]$}
                    \algorithmiccomment{Cleanup Loop}
                    \State $\epsilon \gets \delta_j - f_{j,i+1}$
                    \State Reduce $\mu_{X_j \cup [i], Y_j \cup [i]}$ by $\epsilon$
                    \If {$i+1 \in Y_j$}
                         \State Increase $\mu_{X_j \cup [i], X_j \cup [i+1]}$ and
                         $\mu_{X_j \cup [i+1], Y_j \cup [i+1]}$ by $\epsilon$ each
                    \Else
                         \State Increase $\sigma_{X_j \cup [i+1], Y_j \cup [i]}$ by $\epsilon$
                    \EndIf
              \EndFor
          \EndFor
      \end{algorithmic}
\end{algorithm}

After initialization, the outer loop iterates over $i$ from $0$ to $n-1$. We shall show that
this loop maintains the following {\em outer loop invariant} on the setting of the variables
in $D$, at the end of iteration $i$
\begin{enumerate}
\item The excess  at the vertex $[i+1]$ in $\calL$ is $1$.
\item The excess at every vertex in $\calL$, besides $\emptyset$ and $[i+1]$ is zero.
\item For every  $j\in [k]$, if  $X_j \cup[i+1] \subsetneq Y_j \cup[i+1]$, then  we have
$\mu_{X_j \cup [i+1], Y_j \cup [i+1]} = \delta_j$.
\item $\bm\mu, \bm\sigma \geq \bm 0$.
\end{enumerate}
Note that, if the invariant is satisfied at the end of iteration $i=n-1$, then the resulting $(\bm
\delta,\bm\sigma,\bm\mu)$ will represent a feasible solution for the LP $D$,
which proves the lemma.

We now prove the invariant by induction on $i$. It is
satisfied at $i=0$, where $[0]$ is defined to be the empty set.
For the inductive step, note that the body of the outer loop
has three blocks of inner loops: the forward path loop, the backward path loop, and the
cleanup loop. To extend the inductive hypothesis from $i$ to $i+1$, let $P \in
\mathcal{P}^{i+1}$ be a simple flow path in the graph $G$ that routes $\epsilon$-amount of
flow from $s=\emptyset$ to $t=\{i+1\}$. The first claim examines the effect of the forward
path loop on the variables in $D$.

\begin{claim}
The net effect on an iteration of the forward path loop processing a path $P$, with flow
$\epsilon$, on the setting of the variables in $D$ is:
\begin{itemize}
\item[(a)] The excess  at the vertex $[i]$ in  $\calL$ is reduced by $\epsilon$,
and at the vertex $[i+1]$ is increased by $\epsilon$.
\item[(b)] The excess at every vertex in $\calL$, besides $\emptyset$, $[i]$ and $[i+1]$  is zero.
\item[(c)] For each degree constraint $j \in [k]$ where $(X_j, Y_j) = (A, B) \in P$, and $A
\cup[i] \subsetneq B \cup[i]$, it is the case that $\mu_{[i] \cup X_j, [i] \cup Y_j}$
decreases by  $\epsilon$.
\item[(d)] $\bm\mu, \bm\sigma \geq \bm 0$.
\end{itemize}
\end{claim}
While the algorithm examines each edge $(A,B)$ in $P$ one by one, the real
changes are on the ``lifted'' edge $(A \cup [i], B\cup [i])$.
In particular, the algorithm
examines the path $P[i]$ constructed from $P$ by replacing each edge $(A,B)$ by $(A \cup
[i], B \cup [i])$.
An illustration of the paths $P$, $P[i]$, $P[i+1]$, and flow value settings is shown
in Figure~\ref{fig:forward:backward}.
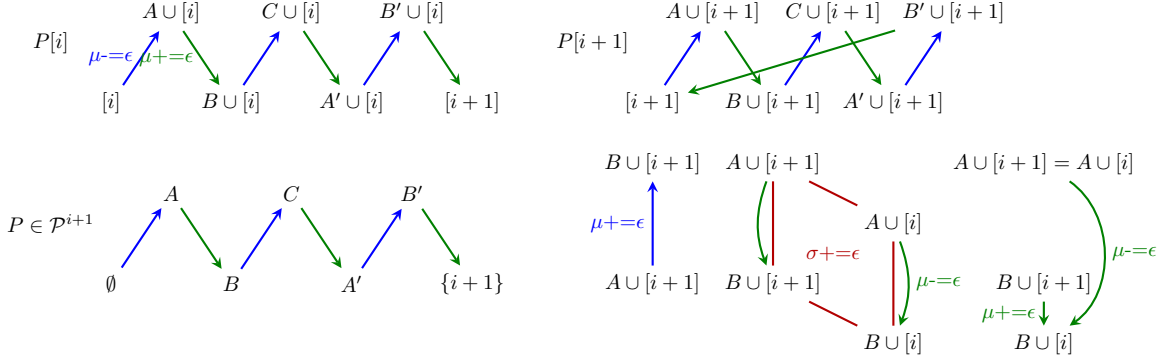
\begin{figure}[th]
\begin{tikzpicture}[scale=0.8, every node/.style={scale=0.8}]
    \node (P) at (-1,0) {$P \in \calP^{i+1}$};
    \node (empty) at (0,-1) {$\emptyset$};
    \node (a) at (1,0.5) {$A$};
    \node (b) at (2,-1) {$B$};
    \node (c) at (3,0.5) {$C$};
    \node (aa) at (4,-1) {$A'$};
    \node (bb) at (5,0.5) {$B'$};
    \node (ip1) at (6,-1) {$\{i+1\}$};
    \draw[>=stealth,->,thick,blue] (empty) to (a);
    \draw[>=stealth,->,thick,DarkGreen] (a) to (b);
    \draw[>=stealth,->,thick,blue] (b) -- (c);
    \draw[>=stealth,->,thick,DarkGreen] (c) -- (aa);
    \draw[>=stealth,->,thick,blue] (aa) -- (bb);
    \draw[>=stealth,->,thick,DarkGreen] (bb) -- (ip1);

    \node (Pi) at (-1,3) {$P[i]$};
    \node (emptyi) at (0,2) {$[i]$};
    \node (ai) at (1,3.5) {$A \cup [i]$};
    \node (bi) at (2,2) {$B \cup [i]$};
    \node (ci) at (3,3.5) {$C \cup [i]$};
    \node (aai) at (4,2) {$A' \cup [i]$};
    \node (bbi) at (5,3.5) {$B' \cup [i]$};
    \node (ip1i) at (6,2) {$[i+1]$};
    \draw[>=stealth,->,thick,blue] (emptyi) to node[left] {$\mu \text{-=} \epsilon$} (ai);
    \draw[>=stealth,->,thick,DarkGreen] (ai) to node[left] {$\mu \text{+=} \epsilon$} (bi);
    \draw[>=stealth,->,thick,blue] (bi) -- (ci);
    \draw[>=stealth,->,thick,DarkGreen] (ci) -- (aai);
    \draw[>=stealth,->,thick,blue] (aai) -- (bbi);
    \draw[>=stealth,->,thick,DarkGreen] (bbi) -- (ip1i);

    \node (Pii)     at (8,3) {$P[i+1]$};
    \node (emptyii) at (9,2) {$[i+1]$};
    \node (aii)     at (10,3.5) {$A \cup [i+1]$};
    \node (bii)     at (11,2) {$B \cup [i+1]$};
    \node (cii)     at (12,3.5) {$C \cup [i+1]$};
    \node (aaii)    at (13,2) {$A' \cup [i+1]$};
    \node (bbii)    at (14,3.5) {$B' \cup [i+1]$};
    \draw[>=stealth,->,thick,blue] (emptyii) to (aii);
    \draw[>=stealth,->,thick,DarkGreen] (aii) to (bii);
    \draw[>=stealth,->,thick,blue] (bii) -- (cii);
    \draw[>=stealth,->,thick,DarkGreen] (cii) -- (aaii);
    \draw[>=stealth,->,thick,blue] (aaii) -- (bbii);
    \draw[>=stealth,->,thick,DarkGreen] (bbii) -- (emptyii);

    \node (Bip1) at (9,1) {$B \cup [i+1]$};
    \node (Aip1) at (9,-1) {$A \cup [i+1]$};
    \draw[>=stealth,->,thick,blue] (Aip1) to node[left] {$\mu \text{+=} \epsilon$} (Bip1);

    \node (Aip12) at (11,1) {$A \cup [i+1]$};
    \node (Bip12) at (11,-1) {$B \cup [i+1]$};
    \draw[>=stealth,->,thick,DarkGreen, bend right=20] (Aip12) to (Bip12);
    \node (Ai2) at (13,0) {$A \cup [i]$};
    \node (Bi2) at (13,-2) {$B \cup [i]$};
    \node[DarkRed] (sigma) at (12,-0.5) {$\sigma \text{+=} \epsilon$};
    \draw[>=stealth,->,thick,DarkGreen, bend left=20] (Ai2) to node[right] {$\mu \text{-=} \epsilon$} (Bi2);
    \draw[thick,DarkRed] (Ai2) -- (Aip12) -- (Bip12) -- (Bi2);
    \draw[thick,DarkRed] (Ai2) to (Bi2);

    \node (AA) at (15.5,1) {$A \cup [i+1] = A \cup [i]$};
    \node (BB) at (15.5,-1) {$B \cup [i+1]$};
    \node (BBB) at (15.5,-2) {$B \cup [i]$};
    \draw[>=stealth,->,thick,DarkGreen, bend left=55] (AA) to node[right] {$\mu \text{-=} \epsilon$} (BBB);
    \draw[>=stealth,->,thick,DarkGreen] (BB) to node[left] {$\mu \text{+=} \epsilon$} (BBB);
\end{tikzpicture}
\caption{Illustrations of Forward and backward passes}
\label{fig:forward:backward}
\end{figure}

Note that $P[i]$ starts from $[i]$ and ends at $[i+1]$; furthermore, if
$A \cup [i] = B \cup [i]$ (i.e. a self-loop) then the edge is not processed. Statements (a)
and (b) follow from the fact that, as we enter a vertex on the path $P[i]$, we either
increase or decrease the excess by $\epsilon$, only to decrease or increase it by $\epsilon$
when processing the very next edge. The only exceptions are the starting vertex $[i]$ which
loses $\epsilon$ excess, and the ending vertex $[i+1]$ which gains $\epsilon$ excess.
Statement (c) follows trivially from line~\ref{line:decrease:mu1} of the algorithm. Part (d)
follows from the induction hypothesis that condition (3) of the outer loop invariant holds
for iteration $i$.

Thus the aggregate effect of the forward path loop (after all paths are processed) is to
increase the excess at $[i+1]$  from 0 to 1, and  to decrease the excess at $[i]$ from from
1 to 0. This establishes the first two conditions of the outer loop invariant for iteration
$i+1$. The purpose of the backward path loop and the cleanup loop is to establish the 3rd
condition of the outer loop invariant for iteration $i+1$. We always maintain $\bm
\sigma,\bm\mu \geq \bm 0$ throughout. The backward path loop increases such a $\mu_{X_j \cup
[i+1], Y_j \cup [i+1]}$  to the flow value $f_{j,i+1}$, where the value of $f_{j,i+1}$ comes
from the optimal solution to the linear program $\Dflow$.  This can still be smaller than
$\delta_j$. The cleanup loop then increases $\mu_{X_j \cup [i+1], Y_j \cup [i+1]}$ to
$\delta_j$, giving the desired property.

The next claim examines the effect of the backward path loop on the variables in $D$.
\begin{claim}
The net effect of an iteration of the backward path loop processing a path $P$ is:
\begin{itemize}
\item[(a)] The excess of all nodes in $\calL$ does not change.
\item[(b)] If there is a degree constraint $j$ where $(X_j, Y_j) = (A,B) \in P$, and $X_j
\cup [i+1] \ne Y_j \cup [i+1]$ then  $\mu_{A \cup [i+1], B \cup [i+1]}$ will increase by
$\epsilon$.
\item[(c)] $\bm\mu, \bm\sigma \geq \bm 0$.
\end{itemize}
\end{claim}
Let $P[i+1]$ be the lifted path constructed from $P$ by replacing each edge $(A, B)$ by an
edge $(A \cup [i+1], B \cup [i+1])$, and removing all self loops. The backward path loop
processes edges in $P[i+1]$. Note that $P[i+1]$ is a closed loop as both the start and the
end are $[i+1]$. Thus it will be sufficient to argue that for each {\em processed} edge $(A
, B)$ in $P$ it is the case that the excess of $A \cup [i+1]$ increases by $\epsilon$, the
excess of $B \cup [i+1]$ decreases by $\epsilon$, the excess of all other nodes does not
change, and  if there is a degree constraint $j$ where $A= X_j$, $B = Y_j$, and $X_j \cup
[i+1] \ne Y_j \cup [i+1]$ then  $\mu_{A \cup [i+1], B \cup [i+1]}$ will increase by
$\epsilon$.

First consider the case that $A \cup [i+1]\subsetneq B\cup [i+1]$, and thus there is a
degree constraint $j$ where $(A,B)=(X_j,Y_j)$. Since we increase $\mu_{A \cup [i+1], B \cup
[i+1]}$ by $\epsilon$ in line~\ref{line:increase:mu2}, it follows that the excess of $A \cup
[i+1]$ increases by $\epsilon$ and the excess of $B \cup [i+1]$ decreases by $\epsilon$. The
excess of all other nodes does not change. Furthermore, (b) and (c) hold for this case.

Second, consider the case that $B \cup [i+1]\subsetneq A\cup [i+1]$.
Decreasing $\mu_{B \cup [i], A \cup [i]}$ by $\epsilon$ in line~\ref{line:increase:mu2} has
the effect of increasing the excess of $A \cup [i]$ and
decreasing the excess of $B \cup [i]$. Figure~\ref{fig:forward:backward} illustrates the
following cases.
If $i+1 \in A$ then, since $A \cup [i] = A\cup [i+1]$ we have an excess increase of
$\epsilon$ at $A \cup [i+1]$. Increasing $\mu_{B \cup [i],B \cup [i+1]}$ by $\epsilon$ in
line~\ref{line:increase:mu3} balances the excess at $B \cup [i]$ and reduces the excess at
$B \cup [i+1]$ by $\epsilon$.
If $i+1 \notin A$, then $(B\cup [i+1]) \cap A \cup [i] = B \cup [i]$ and $(B \cup [i+1]) \cup
(A \cup [i]) = A\cup [i+1]$. In this case, increasing $\sigma_{A \cup [i],B \cup [i+1]}$ by
$\epsilon$ in line~\ref{line:increase:sigma1} balances the excess changes at $A \cup [i]$
and $B \cup [i]$ (due to the decrease in $\mu_{B \cup [i], A \cup [i]}$), increases the
excess at $A \cup [i+1]$, and decreases the excess at $B \cup [i+1]$, as desired.

In summary, after the backward path loop, for each degree constraint $j \in [k]$, where
$X_j \cup[i+1] \subsetneq Y_j \cup[i+1]$, it is the case that $\mu_{ X_j \cup[i+1],  Y_j
\cup [i+1]} =  f_{j,i+1}$, and the value of  $\mu_{ X_j \cup[i],  Y_j \cup [i]}$ is what
remains, which is $\delta_j - f_{j,i+1}$. To bring it up to $\delta_j$, in the cleanup loop
we iterate over each degree constraint $j \in [k]$ where $f_{j,i+1} < \delta_j$ and adjust
the $\bm \mu$ and $\bm \sigma$ variables accordingly. The analysis is analogous to the
analysis of the backward path loop above.
\end{proof}

\begin{ex}[Running example -- constructing  a feasible solution for $D$]
Order the nodes as
$$(1) \; a, (2) \; b, (3) \; c, (4) \; d$$
Briefly, we  discuss iterations $i = 0$
and  $i=3$.  Recall that the optimal solution sets $\delta_1, \delta_3$ and $\delta_4$ all
to one and initially these variables are one in the lattice.  Before the outer loop,
$\mu_{\emptyset, \{a,b\}}, \mu_{\emptyset, \{a,c\}}$ and $\mu_{\{a\}, \{a,d\}}$ are all also
set to one.

\emph{First Iteration:} Consider the outer loop where $i=0$.  In the auxiliary
graph, the flow to $\{a\}$ is a single path  $\emptyset, \{ab\}, \{ab\},\{a\}$ .  During the
forward path loop, the variable $\mu_{\emptyset, \{a,b\}}$ decreases by $1$ and
$\mu_{\{a,b\}, \{a\}}$ increases by one.  In the backward path loop, $\mu_{\{a\},\{a,b\}}$
decreases by one and then increases again by one. During the cleanup loop, the degree
constraint $(\emptyset, \{a,c\})$ results in $\mu_{\emptyset, \{a,c\}}$ decreases by one.
Then $\mu_{\emptyset, \{a\}}$ and $\mu_{\{a\}, \{a,c\}}$  increase by one. Next consider the
degree constraint, $(\{a\}, \{a,d\})$.  The variable $\mu_{\{a\}, \{a,d\}}$ decreases by one
and then immediately increase by one again.

\emph{Last Iteration:}  In this case $[i]$ is $\{a,b,c\}$.  Just before the outer
loop, $\mu_{\{a,b,c\}, \{a,b,c,d\}}$ is set to one, corresponding to  $\delta_4$ and the
excess at vertex $\{a,b,c\}$ is one.   In the auxiliary graph, the flow to $\{d\}$ a single
path  $\emptyset, \{ab\},\{ab\},\{a\},\{a\},\{ad\},\{a,d\},\{d\}$.  In the forward pass,
nothing changes when processing the edges $\emptyset,\{ab\}$ and
$\{ab\},\{a\}$.\footnote{There is no change  because  $\emptyset\cup [i] = \{ab\} \cup [i] =
\{a\} \cup [i]$.}  When $\{a\},\{ad\}$ is processed,  $\mu_{\{a,b,c\}, \{a,b,c,d\}}$
decreases by $1$. Nothing changes when processing the last edge of the path and $[i+1]$ is
the universe $\{a,b,c,d\}$ so nothing changes in the backward or cleanup loops. This
effectively gets a flow of $1$ to the universe.
\end{ex}

\subsection{Simple degree constraints and the normal bound}
\label{subsec:simple:normal}

An interesting consequence of our approach is the following result, first
proved in~\cite{DBLP:conf/pods/KhamisK0S20}.

\begin{prop}
   If the input degree constraints are simple, then $\dc[\Gamma_n]=\dc[\Nor_n]$.
\end{prop}
\begin{proof}
The dual LP of $\Dsimple$ \eqref{eqn:DS} is
\begin{align}
    P_{\bm \delta}^{\sf simple}: && \max  &\qquad \sum_{V\subset [n]} \lambda_V && \label{eqn:PS} \\
         && \text{s.t.} & \qquad
        \sum_{V \subset [n]: \  X_i \subseteq V,  Y_i \not\subseteq V}
    \lambda_V \le c_i, & \forall i \in [k] \nonumber
\end{align}
From the results of Section~\ref{subsec:simple}, we know that the LP $P_{\bm \delta}^{\sf
simple}$ shown in~\eqref{eqn:PS} has optimal value equal $\dc[\Gamma_n]$. We prove that this
LP is the same as $\dc[\Nor_n]$. Recall that $h \in \Nor_n$ iff $h = \sum_V \lambda_V s_V$,
for some non-negative $\lambda_V$, and where $s_V$ is the step function defined by $s_V(X)
:= \bm 1_{X \not\subseteq V}$ (the indicator function for the event $X \subseteq V$). It
follows that~\eqref{eqn:PS} is {\em exactly} $\dc[\Nor_n]$, because
\begin{align*}
h(X) &= \sum_{V: X \not\subseteq V} \lambda_V &&
h([n]) = \sum_{V \subset [n]} \lambda_V &&
h(Y)-h(X) = \sum_{V: X \subseteq V, Y \not\subseteq V} \lambda_V. && X \subseteq Y
\end{align*}
In terms of weighted coverage function~\cite{DBLP:books/cu/p/0001G14},
$\lambda_V$ is the weight of the vertex connected to vertices in $V$ in the standard
bipartite representation (WLOG we assume that there is only one such vertex).
\end{proof}

\section{Computing a polynomial sized proof-sequence for simple instances}
\label{sec:ps:simple}
This section proves Theorem~\ref{thm:main2} by constructing, from a feasible solution to the
linear program $\Dflow$, a proof sequence for the Shannon-flow inequality $h([n] |
\emptyset) \leq \sum_{j \in [k]} \delta_j \cdot h(Y|X)$. The construction is given in
Algorithm~\ref{algo:delta:to:ps}, whose execution traces the construction of the $\sigma$
and $\mu$ variables from the $\delta$ variables in Algorithm~\ref{algo:delta:to:sigma:mu}.
We now show in Lemma \ref{lem:alg2invariant} an invariant of
Algorithm~\ref{algo:delta:to:ps} that is sufficient to show that the proof sequence is
correct. To see this note that this invariant implies that in the end the excess at $[n]$ is
$1$.

\begin{algorithm}[h]
     \caption{Constructing a proof sequence from $\bm\delta$ feasible to $\Dflow$}
     \label{algo:delta:to:ps}

     \begin{algorithmic}[1]
          \State (We write $h(X)$ to mean $h(X \mid \emptyset)$ for short)
          \For {$i = 0$ to $n-1$}
              \algorithmiccomment{Outer Loop}
              \State Let $\calP^{i+1}$ be the collection of simple flow paths
              routing $1$ unit of flow from $s=\emptyset$ to $t=\{i+1\}$
              \For {each path $P \in \calP^{i+1}$ with flow value $\epsilon$}
                    \algorithmiccomment{Forward Path Loop}
                    \For {each edge $(A,B) \in P$ from $\emptyset$ to $\{i+1\}$}
                         \If {$A \cup [i] \subsetneq B \cup [i]$}
                              \State $\epsilon$-Compose: $h(A\cup [i])+ h(B\cup[i]  \mid  A\cup [i]) \to h(B\cup [i])$
                              \label{line:comp:1}
                         \ElsIf {$B \cup [i] \subsetneq A \cup [i]$}
                              \State $\epsilon$-Decompose: $h(A \cup [i]) \to h(B\cup [i])+ h(A\cup[i]  \mid  B\cup [i])$
                              \label{line:decomp:1}
                         \EndIf
                    \EndFor
              \EndFor
              \For {each path $P \in \calP^{i+1}$ with flow value $\epsilon$}
                    \algorithmiccomment{Backward Path Loop}
                    \For {each edge $(A,B) \in P$ from $\{i+1\}$ back to $\emptyset$}
                         \If {$A \cup [i+1] \subsetneq B \cup [i+1]$}
                              \State $\epsilon$-Decompose: $h(B \cup [i+1]) \to h(A\cup [i+1])+ h(B\cup[i+1]  \mid  A\cup [i+1])$
                              \label{line:decomp:2}
                         \ElsIf {$B \cup [i+1] \subsetneq A \cup [i+1]$}
                              \If {$i+1 \in A$}
                                   \algorithmiccomment{This means $A\cup [i] = A\cup [i+1]$}
                                   \If {$B \cup [i+1] \neq B\cup [i]$}
                                        \State $\epsilon$-Decompose: $h(B \cup [i+1]) \to h(B\cup [i])+ h(B\cup[i+1]  \mid  B\cup [i])$
                                   \label{line:decomp:3}
                                   \EndIf
                                   \State $\epsilon$-Compose: $h(B\cup [i])+ h(A\cup[i]  \mid  B\cup [i]) \to h(A\cup [i]) = h(A \cup [i+1])$
                                   \label{line:comp:2}
                              \Else
                                   \State $\epsilon$-Submodularity: $h(A\cup [i]  \mid  B \cup [i]) \to h(A\cup [i+1]  \mid  B \cup [i+1])$
                                   \label{line:submod:1}
                                   \State $\epsilon$-Compose: $h(B\cup [i+1])+ h(A\cup[i+1]  \mid  B\cup [i+1]) \to h(A\cup [i+1])$
                                   \label{line:comp:3}
                              \EndIf
                         \EndIf
                    \EndFor
              \EndFor
              \For {each $j \in [k]$, where $X_j \cup [i+1] \subsetneq Y_j \cup [i+1]$} and $[i+1] \not\subset X$
                    \algorithmiccomment{Cleanup Loop}
                    \State $\epsilon \gets \delta_j - f_{j,i+1}$
                    \If {$i+1 \in Y_j$}
                        \State $\epsilon$-Decompose:
                              $h(Y_j\cup [i]  \mid  X_j \cup [i]) \to h(Y_j\cup [i]  \mid  X_j \cup [i+1]) + h(X_j \cup [i+1]\mid X_j \cup [i]) $
                        \State $\epsilon$-Monotonicity:
                              $h(X_j \cup [i+1]\mid X_j \cup [i]) \to 0$
                    \Else
                         \State $\epsilon$-Submodularity: $h(X_j\cup [i]  \mid  Y_j \cup [i]) \to h(X_j\cup [i+1]  \mid  Y_j \cup [i+1])$
                    \EndIf
              \EndFor
          \EndFor
      \end{algorithmic}
\end{algorithm}

\begin{lmm}
\label{lem:alg2invariant}
    Starting from the sum $\sum_{j \in [k]} \delta_j \cdot h(Y_j|X_j)$, the proof sequence
    constructed in Algorithm~\ref{algo:delta:to:ps} satisfies the following invariants. If
    we were to run Algorithm~\ref{algo:delta:to:sigma:mu} in lock step with
    Algorithm~\ref{algo:delta:to:ps}, then after every edge $(A,B)$ is processed  the
    following holds:
    \begin{itemize}
        \item[(a)] The coefficient of $h(Y \mid \emptyset)$ is the excess at $Y$, for every
        $Y \subseteq [n]$.
        \item[(b)] The coefficient of $h(Y \mid X)$ is exactly $\mu_{X,Y}$, for every
        $\emptyset \neq X \subseteq Y \subseteq [n]$.
    \end{itemize}
\end{lmm}
\begin{proof}
    We prove the lemma by induction. Initially, when we initialize $\bm \mu$ in
    line~\ref{line:initialize:mu} of Algorithm~\ref{algo:delta:to:sigma:mu}, the invariants
    hold trivially. We verify that the invariant holds after each proof step. Please also
    refer to the proof of Lemma~\ref{lem:simpleliftA2} as we need to run the two proofs in
    parallel.

    In the forward path loop, we traverse edges $(A\cup [i],B \cup [i])$ of $P[i]$ from
    $\emptyset$ to $\{i+1\}$. When $A \cup [i] \subsetneq B \cup [i]$, in
    Algorithm~\ref{algo:delta:to:sigma:mu} we decrease $\mu_{A \cup [i], B \cup [i]}$ by
    $\epsilon$, increase the excess at $B\cup [i]$ by $\epsilon$, reduce the excess at
    $A\cup [i]$ by $\epsilon$, which correspond precisely to $\epsilon$-composing $h(A\cup
    [i])+ h(B\cup[i]  \mid  A\cup [i]) \to h(B\cup [i])$ in
    Algorithm~\ref{algo:delta:to:ps}. The case when $B \cup [i] \subsetneq A \cup [i]$ is
    the converse.

    In the backward path loop, we traverse edges $(A\cup[i+1],B\cup[i+1])$ of $P[i+1]$ from
    ${i+1}$ back to $\emptyset$. For $A \cup [i+1] \subsetneq B \cup [i+1]$, in
    Algorithm~\ref{algo:delta:to:sigma:mu} we increase $\mu_{A \cup [i+1], B \cup [i+1]}$,
    decrease the excess at $B\cup [i+1]$, and increase the excess at $A\cup [i+1]$ by
    $\epsilon$. These correspond precisely to $\epsilon$-decomposing $h(B \cup [i+1]) \to
    h(A\cup [i+1])+ h(B\cup[i+1]  \mid  A\cup [i+1])$ in Algorithm~\ref{algo:delta:to:ps}.
    When $B \cup [i+1] \subsetneq A \cup [i+1]$, there are two cases. For $i+1 \in A$,
    Algorithm~\ref{algo:delta:to:sigma:mu} decreases $\mu_{B \cup [i], A \cup [i]}$
    increases $\mu_{B \cup [i],B \cup [i+1]}$, reduces the excess at $B \cup [i+1]$, and
    increases the excess at $A \cup [i+1]$ by $\epsilon$. This corresponds to the
    $\epsilon$-decomposing step $h(B \cup [i+1]) \to h(B\cup [i])+ h(B\cup[i+1]  \mid  B\cup
    [i])$ and the $\epsilon$-composing step $h(B\cup [i])+ h(A\cup[i]  \mid  B\cup [i]) \to
    h(A\cup [i]) = h(A \cup [i+1])$ in Algorithm~\ref{algo:delta:to:ps}. For $i+1 \notin A$,
    Algorithm~\ref{algo:delta:to:sigma:mu} increases $\sigma_{A \cup [i],B \cup [i+1]}$,
    reduces $\mu_{B\cup[i],A\cup[i]}$, reduces the excess at $B \cup [i+1]$, and increases
    the excess at $A \cup [i+1]$ by $\epsilon$. This corresponds to the
    $\epsilon$-submodularity and $\epsilon$-compose steps as shown in
    lines~\ref{line:submod:1} and ~\ref{line:comp:3} of Algorithm~\ref{algo:delta:to:ps}.

    Lastly, the cleanup loop is self-explanatory, designed to maintain the invariants.
\end{proof}

Note that as the graph $G$ has $O(k + n)$ vertices and $O(k n)$ edges, one can assume
without loss of generality that the cardinality of  $\calP^{i+1}$ is $O(kn)$, and the length
of each path   $P \in \calP^{t}$ is $O(k+n)$. As each edge $e \in P$ introduces $O(1)$ steps
to the proof sequence when $P$ is processed, and there are at most $n$ choices for the sink
$t$, we can conclude that the length of the resulting proof sequence is $O( (k+n) k n^2)$,
or equivalently $O(  k^2 n^2 + kn^3)$.

\begin{ex}[Running example]
Algorithm~\ref{algo:delta:to:ps} yields the following proof sequence:
\begin{eqnarray*}
&&h(\emptyset)+ h(ab| \emptyset) + h(ac| \emptyset)+ h(ad | a)\\
&=&  h(ab) +h(ac|\emptyset)+ h(ad | a)  \qquad \mbox{[Forward pass $1$-Compose: $h(ab) = h(\emptyset) + h(ab|\emptyset)$]}\\
&=& h(a) + h(ab|a)  +h(ac|\emptyset)+ h(ad | a)  \qquad \mbox{[Forward pass $1$-Decompose: $h(ab) = h(a) + h(ab|a)$]}\\
&=& h(ab) + h(ac| \emptyset)+ h(ad | a)  \qquad \mbox{[Backward pass $1$-Compose: $ h(a) + h(ab|a) = h(ab) $]}\\
&=& h(a) + h(ab|a) + h(ac| \emptyset)+ h(ad | a)  \qquad \mbox{[Backward pass $1$-Decompose: $ h(ab) = h(a)+ h(ab|a) $]}\\
%
&\geq& h(a) + h(ab|a) +h(ac|a) + h(a|\emptyset)+ h(ad | a)  \;\;\;\mbox{[Clean-up: $1$-Decompose: $  h(ac|\emptyset) \geq h(ac|a)+  h(a|\emptyset)$]}\\
&\geq& h(a) + h(ab|a) +h(ac|a)+ h(ad | a)  \qquad \mbox{[Clean-up: $1$-Monotonicity: $  h(a|\emptyset) \geq 0 $]}\\
&=& h(ab) +h(ac|a)+ h(ad | a)  \qquad \mbox{[Forward pass $1$-Compose: $h(ab) = h(a) + h(ab|a)$]}\\
&\geq& h(ab) +h(abc|ab)+ h(ad | a)  \qquad  \mbox{[Clean-up: $1$-Submodularity: $  h(ac|a) \geq h(abc|ab) $]}\\
&\geq& h(ab) + h(abc|ab)+ h(abd | ab)  \qquad  \mbox{[Clean-up: $1$-Submodularity: $  h(ad|a) \geq h(abd|ab) $]}\\
&=& h(abc)+  h(abd | ab)  \qquad  \mbox{[Forward pass $1$-Compose: $h(abc) = h(ab) + h(abc|ab)$]}\\
&\geq& h(abc)+  h(abcd | abc)  \qquad \mbox{[Clean-up: $1$-Submodularity: $  h(abd | ab)\geq h(abcd|abc) $]}\\
&\geq& h(abcd) \qquad \mbox{[Forward pass $1$-Compose: $h(abcd) = h(abc) + h(abcd|abc)$]}\\
\end{eqnarray*}

\end{ex}

\section{Lower bounds}
\label{sec:lower-bound}
In this section we present three classes of seemingly easy instances which turn out to be as
hard as general instances. Lemmas~\ref{lmm:acyclicplussimple-main},
~\ref{lmm:2-3-reduction-main}, and ~\ref{lmm:simple-fd-reduction-main} combined would imply
Theorem~\ref{thm:not:much:better}.

\begin{lmm}
\label{lmm:acyclicplussimple-main}
For the problem of computing the polymatroid bound, an arbitrary instance can be converted
into another instance in polynomial time without changing the bound, where the degree
constraints is the union of a set of acyclic degree constraints and a set of simple degree
constraints (in fact functional dependencies) Further, each FD contains exactly two
variables.
\end{lmm}
\begin{proof}
We first describe the reduction. Suppose we are given an arbitrary instance $I$ consisting
of the universe $U := [n]$ and a set of degree constraints, $\mathrm{DC}$. The new instance
$I'$ has $U' := \bigcup_{i \in [n]}\{x_i, y_i\}$ as universe where $x_i$ and $y_i$ are
distinct copies of $i$ and the following set $\mathrm{DC'}$ of constraints. For each $i \in
[n]$, we first add the following {\em simple} functional dependencies to $\mathrm{DC'}$:
\begin{align*}
    (\{x_i\}, \{x_i, y_i\}, 0) && (\{y_i\}, \{x_i, y_i\}, 0)
\end{align*}
Then for each constraint $(A, B, c) \in \mathrm{DC}$, we create a new constraint
$(A', B', c)$ by replacing each $i \in A$ with $x_i$ and each $j \in B$ with $y_j$, and add
it to $\mathrm{DC'}$. By construction these degree constraints are from $\{x_1, x_2, \cdots,
x_n\}$ to $\{y_1, y_2, \cdots, y_n\}$ and therefore are {\em acyclic}.

The following simple observation states that $x_i$ and $y_i$ are indistinguishable in
computing the polymatroid bound for $I'$.

\begin{claim}
\label{claim:g-properties}
Let $g$ be a submodular function that satisfies $\mathrm{DC'}$  of $I'$. For any $i \in [n]$
and any $B \subseteq U'$ such that $x_i, y_i \not \in B$, we have $g(B\cup \{x_i\}) = g(B
\cup \{y_i\}) = g(B \cup \{x_i,y_i\}).$
\end{claim}
By monotonicity and submodularity of $g$, and an FD in $\mathrm{DC'}$
involving $x_i, y_i$, we have:
\begin{align*}
0 \leq g(B\cup \{x_i,y_i\}) - g(B\cup \{x_i\}) \leq g(\{x_i,y_i\}) - g(\{x_i\}) \leq 0.
\end{align*}
This proves $g(B\cup \{x_i,y_i\}) = g(B\cup \{x_i\})$. The other equality $g(B\cup
\{x_i,y_i\}) = g(B\cup \{y_i\})$ is established analogously.

Henceforth, we will show the following to complete the proof of
Lemma~\ref{lmm:acyclicplussimple-main}:
\begin{itemize}
\item[$\Rightarrow$] Given a monotone submodular function $f$ achieving the optimum
polymatroid bound for $I$, we can create a monotone submodular function $g$ for $I'$ such
that $f(U) = g(U')$.
\item[$\Leftarrow$] Conversely, given a monotone submodular function $g$ achieving the
optimum polymatroid bound for $I'$, we can create a monotone submodular function $f$ for $I$
such that $f(U) = g(U')$.
\end{itemize}

We start with the forward direction.  Define $h : 2^{U'} \rightarrow 2^{U}$ as follows: for
any $B \subseteq U'$, we have
\[ h(B) := \{i \in [n] \mid x_i \in B \text{ or } y_i \in B\}.  \]
and set $g(B) := f(h(B))$.

\begin{claim}
\label{claim:h-properties}
For any $A, B \subseteq U'$, $h(A) \cup h(B) = h(A \cup B)$ and $h(A) \cap h(B) \supseteq
h(A \cap B)$.
\end{claim}
The first statement follows because  if $x_i$ or $y_i$ is in any of $A$ and $B$, it is also
in $A\cup B$. The second statement follows because if $i  \in h(A \cap B)$, we have $x_i \in
A \cap B$ or $y_i \in A \cap B$ and in both cases, we have $i \in h(A) \cap h(B)$.

At the universe $U$ and $U'$: by definition we have $g(U') = f(h(U')) = f(U)$. Therefore, we
only need to show $g$ is monotone and submodular. Showing monotonicity is trivial and is
left as an easy exercise. We can show that $g$ is submodular as follows. For any $A, B
\subseteq U'$, we have,
\begin{align*}
    g(A) + g(B) &= f(h(A)) + f(h(B)) \\
    &\geq f(h(A)\cup h(B)) + f(h(A)\cap h(B)) \\
    &\geq f(h(A \cup B)) + f(h(A \cap B)) \\
    &= g(A\cup B) + g(A\cap B),
\end{align*}
where the first inequality follows from $f$'s submodularity and the second from $f$'s
monotonicity and Claim~\ref{claim:h-properties}. Thus we have shown the first direction.

To show the backward direction, define
\[
f(A) := g(\{x_i \mid i \in A \})
\]
By definition, we have $f(U) = g(\{x_1, x_2, \cdots, x_n\})$. Further, by repeatedly
applying Claim~\ref{claim:g-properties}, we have $g(\{x_1, x_2, \cdots, x_n\}) = g(U')$. Thus
we have shown $f(U) = g(U')$. Further,  $f$ is essentially identical to $g$ restricted to
$\{x_1, x_2, \ldots, x_n\}$. Thus, $f$ inherits $g$'s monotonicity and sumodularity.

This completes the proof of Lemma~\ref{lmm:acyclicplussimple-main}.
\end{proof}

\begin{lmm}
\label{lmm:2-3-reduction-main}
There is a polynomial-time reduction from a general instance to an instance preserving the
polymatroid bound, where for each degree constraint $(X, Y, c)$ we have $|Y| \leq 3$ and
$|X| \leq 2$. Further, the new instance satisfies the following:
\begin{itemize}
\item If $|Y| = 3$, then $|X|=2$ and $c = 0$.
\item If $|Y| = 2$, then $|X|=1$.
\end{itemize}
\end{lmm}
\newcommand{\opt}{\textsf{opt}}
\begin{proof}[Proof of Lemma~\ref{lmm:2-3-reduction-main}]
The high-level idea is to repeatedly replace two variables with a new variable in a degree
constraint. We first discuss how to choose two variables to combine. Assume there is a
degree constraint $(X, Y, c)$ where $|X| > 2$. Then we combine an arbitrary pair of elements
in $X$. If $|X| \leq 2$ for all degree constraints, and there is a degree constraint where
$|Y| > 3$, we combine arbitrary two variables in $Y \setminus X$. If there is a degree
constraint $(X, Y, c)$ where $|X| = 2$, $|Y| = 3$ and $c > 0$, we combine the two variables
in $X$. It is important to note that we make this replacement in only one degree constraint
in each iteration.

Assume that we are to combine variables $x,y \in [n]$ into a new variable $z \notin [n]$
(say $z=n+1$) in a degree constraint $(X, Y, c) \in \mathrm{DC}$. Then, we create $(X', Y',
c)$ and add it to $\mathrm{DC'}$ where
\begin{align*}
    (X', Y', c) &:=
    \begin{cases}
        (X \setminus \{x,y\} \cup \{z\}, Y \setminus \{x,y\} \cup \{z\}, c)
            & \text{if } \{x,y\} \subseteq X \subset Y \\
        (X, Y \setminus \{x,y\} \cup \{z\}, c)
            & \text{if } \{x,y\} \subseteq Y \setminus X
    \end{cases}
\end{align*}
Further, we add functional dependencies $(\{x,y\}, \{z, x,y\},0 )$, $(\{z\}, \{z,
x\},0)$ and $(\{z\}, \{z,y\},0)$ to $\mathrm{DC'}$, which we call {\em consistency
constraints}. Intuitively, consistency constraints is to enforce the fact that variable $z$
and the tuple of variables $(x,y)$ are equivalent.
The other constraints are called {\em non-trivial constraints}.

We show that the reduction process terminates by showing that each iteration reduces
a potential. Define $m(X, Y,
c) := |X| + |Y|$ for a non-trivial constraint $(X, Y, c)$. The potential is defined as the
sum of $m(X, Y, c)$ over all non-trivial constraints.  Observe that in each iteration,
either $m(X, Y, c) > m(X', Y', c)$, or $|X'| = 1$ and $|Y'| = 2$.  In the latter case, the
resulting non-trivial constraint $(X', Y', c)$ doesn't change in the subsequent iterations.
In the former case the potential decreases. Further, initially the potential is at most $2n
|\mathrm{DC}|$ and the number of non-trivial constraints never increases, where
$\mathrm{DC}$ is the set of degree constraints initially given. Therefore, the reduction
terminates in a polynomial number of iterations. It is straightforward to see that we
only have degree constraints of the forms that are stated in
Lemma~\ref{lmm:2-3-reduction-main} at the end of the reduction.

Next, we prove that the reduction preserves the polymatroid bound. We consider one iteration
where a non-trivial constraint $(X, Y, c)$ is replaced according to the reduction described
above. Let $\opt$ and $\opt'$ be the polymatroid bounds before and after performing the
iteration respectively. Let $\mathrm{DC}$ and $\mathrm{DC'}$ be the sets of the degree
constraints before and after the iteration respectively.

We first show $\opt \geq \opt'$. Let $g: \{z\} \cup [n] \rightarrow [0, \infty)$ be a
polymatroid function that achieves $\opt'$ subject to $\mathrm{DC'}$.  Define $f:
[n] \rightarrow [0, \infty)$ such that $f(A)  = g(A)$ for all $A \subseteq [n]$. It is
immediate that $f$ is a polymatroid because it is a restriction of $f$ onto $[n]$.
We show below that it satisfies the degree constraints in $\mathrm{DC}$ and that
$f([n]) = g([n] \cup \{z\}) = \opt'$.
We start with a claim.

\begin{claim}
    \label{claim:1-2-equal}
	For any set $Z$, $g(Z \cup \{z\} )= g(Z \cup \{z, x,y\}) = g(Z \cup \{x,y\})$.
\end{claim}
From the consistency constraint
$(\{x,y\}, \{z, x,y\}, 0)$ and the monotonicity of $g$, we have
$0 \leq g(\{z, x,y\}) - g(\{x,y\}) \leq 0$, which means
$g(\{x,y,z\}) = g(\{x,y\})$.
Similarly, from the consistency constraints $(\{z\}, \{z, x\}, 0)$ and $(\{z\}, \{z,y\}, 0)$
and the monotonicity of $g$, we have $g(\{z\}) = g(\{z, x\}) = g(\{z,y\})$.
Then, from $g$'s submodularity and monotonicity, we have
\[
g(\{x,y,z\}) \leq g(\{z,x\}) + g(\{z,y\}) - g(\{z\}) = g(\{z\}) \leq g(\{x,y,z\}),
\]
This implies $g(\{z\}) = g(\{x,y,z\})$.

Now, given two sets $A, B$ such that $g(A) = g(A \cup B)$, then for any set $Z$ we have
\[
g(Z \cup A \cup B)
    \leq g(Z \cup A) + g(A \cup B) - g(A \cup (Z \cap B))
    \leq g(Z \cup A) + g(A \cup B) - g(A)
    = g(Z \cup A)
    \leq g(Z \cup A \cup B).
\]
Therefore, $g(Z \cup A) = g(Z \cup A \cup B)$. Applying this fact with
$A = \{z\}$ and $B = \{x,y,z\}$, we have $g(Z \cup \{z\}) = g(Z \cup \{x,y,z\})$.
Similarly, $g(Z \cup \{x,y\}) = g(Z \cup \{x,y,z\})$.
Claim~\ref{claim:1-2-equal} is thus proved.

We now check if $f$ satisfies $\mathrm{DC}$. Because we only replaced $(X, Y, c) \in
\mathrm{DC}$, we only need to show that $f$ satisfies it.
Note that, if $\{x,y\} \subseteq Z$, then from Claim~\ref{claim:1-2-equal}, we have
\[
f(Z) = g(Z) = g(Z \cup \{x,y\}) = g(Z \cup \{z\})
= g( (Z \setminus \{x,y\}) \cup \{x, y, z\} )
= g( (Z \setminus \{x,y\}) \cup \{z\} )
\]
We need to consider two case:
\begin{itemize}
	\item When $\{x,y\}  \subseteq X \subseteq Y$. Then,
	\[ f(Y) - f(X)  = g( Y \cup \{z\} \setminus \{x,y\}) - g(X \cup \{z\} \setminus \{x,y\})
    = g(Y') - g(X') \leq c \]
    The second inequality follows from the fact that $g$ satisfies $\mathrm{DC'}$.
	\item When $\{x,y\}  \subseteq Y \setminus X$. In this case, $X' = X$ and $Y' = Y \cup
	\{z\} \setminus \{x,y\}$; thus
	\[ f(Y) - f(X)  = g( Y \cup \{z\} \setminus \{x,y\}) - g(X) = g(Y') - g(X') \leq c \]
\end{itemize}

Finally, $f([n]) = g([n]) = g([n] \cup \{z\}) = \opt'$ due to Claim~\ref{claim:1-2-equal}.
Since we have shown $f$ is a feasible solution for $\mathrm{DC}$, we have $\opt \geq
f([n])$. Thus, we have $\opt \geq \opt'$ as desired.

We now show $\opt \leq \opt'$. Given $f$ that achieves $\opt$ subject to $\mathrm{DC}$, we
construct $g: \{z\} \cup [n] \rightarrow [0, \infty)$ as follows:
\begin{equation}
	g(A) :=
	\begin{cases}
		f(A) & \mbox{if } z \not \in A \\
		f(A \setminus \{z\} \cup \{x,y\}) & \textnormal{otherwise}
 	\end{cases}
\end{equation}

We first verify that $g$ is monotone. Consider $A \subseteq B \subseteq \{z\} \cup [n]$. If
$z \not \in A$ and $z \not \in B$, or  $z \in A$ and $z \in B$, it is easy to see that is
the case. So, assume $z \not \in A$ but $z \in B$. By definition of $g$, it suffices show
$f(A) \leq f(B \setminus \{z\} \cup \{x,y\})$, which follows from $f$'s monotonicity:
Since $z \not \in A$ and $A \subseteq B$,
we have $A \subseteq B \setminus \{z\} \cup \{x, y\}$.

Secondly we show that $g$ is submodular. So, we want to show that $g(A) + g(B) \geq g(A \cup
B) + g(A \cap B)$ for all $A, B \subseteq \{z\} \cup [n]$.

\begin{itemize}
\item When $z \not \in A$ and $z \not \in B$. This case is trivial as $g$ will have the same
value as $f$ for all subsets we're considering.
\item When $z \in A$ and $z \in B$. We need to check if $f(A \setminus \{z\} \cup \{x,
	y\}) + f(B \setminus \{z\} \cup \{x,y\})  \geq f(A \cup B \setminus \{z\} \cup \{x,
	y\}) + f(A \cap B \setminus \{z\} \cup \{x,y\})$, which follows from $f$'s
	submodularity. More concretely,  we  set $A' = A \setminus \{z\} \cup \{x,y\}$ and
	$B' = B \setminus \{z\} \cup \{x,y\}$ and use $f(A') + f(B') \geq f(A' \cup B') +
	f(A' \cap B')$.
\item When $z \in A$ and $z \not \in B$ (this is symmetric to $z \not \in A$ and $z \in B$).
We have
\begin{align*}
    g(A) + g(B) &= f(A \setminus \{z\} \cup \{x,y\}) + f(B) \\
(\text{submodularity of $f$}) &\geq f((A \setminus \{z\} \cup \{x,y\}) \cup B) + f((A \setminus \{z\} \cup \{x,y\}) \cap B) \\
&= f(((A \cup B) \setminus \{z\}) \cup \{x,y\}) + f((A \setminus \{z\} \cup \{x,y\}) \cap B) \\
&= g(A \cup B) + f((A \cup \{x,y\}) \cap B) \\
(\text{monotonicity of $f$}) &\geq g(A \cup B) + f(A \cap B)\\
&= g(A \cup B) + g(A \cap B)
\end{align*}
\end{itemize}

Thirdly, we show that $g$ satisfies $\mathrm{DC'}$. Suppose we replaced a non-trivial
constraint $(X, Y, c)$ with $(X', Y', c)$. We show $g(Y') - g(X') \leq c$ by showing $f(Y)
= g(Y')$ and $f(X) = g(X')$. Both cases are symmetric, so we only show $f(X') = g(X)$. If
$z \not \in X'$, then clearly we have $g(X') = f(X)$ since $X' = X$. If $z \in X'$, then it
must be the case that $X'  = X \setminus \{x,y\} \cup \{z\}$. By definition of $g$, we
have $g(X') = f(X'  \setminus \{z\} \cup \{x,y\})  = f(X)$ since $X'  \setminus \{z\}
\cup \{x,y\} = X$.

Now we also need to check $g$ satisfies the consistency constraints we created. So we show
\begin{itemize}
\item $g(\{z, x,y\}) \leq g(\{z\})$. Note $g(\{z, x,y\}) = f(\{x,y\}) = g(\{z\})$
by definition of $g$. Due to $g$'s monotonicity we have already shown, we have $g(\{z,
x\}) \leq g(\{z\})$ and $g(\{z,y\}) \leq g(\{z\})$.
\item $g(\{z, x,y\}) \leq g(\{x,y\})$. Both sides are equal to $ f(\{x,y\})$ by
definition of $g$.
\end{itemize}

Finally, we have $g(\{z\} \cup [n]) = f([n])$. Since $g$ is a monotone submodular function
satisfying $\mathrm{DC}$, we have $\opt' \geq \opt$ as desired.

This completes the proof of Lemma~\ref{lmm:2-3-reduction-main}.
\end{proof}

\begin{lmm}
\label{lmm:simple-fd-reduction-main}
There is a polynomial-time reduction from a general instance to an instance consisting of a
set of simple degree constraints and a set of functional dependency constraints, while
preserving the polymatroid bound.
\end{lmm}

The proof of
Lemma~\ref{lmm:simple-fd-reduction-main} is  similar to (and much simpler than) that of
Lemma~\ref{lmm:2-3-reduction-main} and thus is omitted. The high-level idea is as follows.
We repeat the following:
If there exist $x_1 \neq x_2 \in X$ for some degree constraint $(X, Y, c)$, replace the
constraint with $(X \setminus \{x_1, x_2 \} \cup \{x_{12}\}, Y \setminus \{x_1, x_2 \} \cup
\{x_{12}\}, c)$ and analogously update other constraints; and ii) For consistency, we add
functional dependencies $(\{x_1, x_2\}, \{x_1, x_2, x_{12}\}, 0)$ and $(\{x_{12}\}, \{x_1,
x_2, x_{12}\}, 0)$.

\section{Hardness of computing normal bounds}
\label{sec:hardness1}
\newcommand{\po}{\Delta}
\newcommand{\cS}{\mathcal{S}}
\newcommand{\lpn}{\textsf{LP}_{\textrm{normal}}}

This section sketches the proof that the normal bound $\mathrm{DC}[\Nor_n]$ cannot be solved
in polynomial time unless P = NP, i.e., Theorem~\ref{thm:normalhard}.

Recall that  {\em Normal} functions~~\cite{DBLP:conf/pods/KhamisK0S20,csma} (also called
{\em weighted coverage functions}, or {\em entropic functions with non-negative mutual
information}), are defined as follows. For every $V \subsetneq [n]$, a {\em step function}
$s_V : 2^{[n]} \to \R_+$ is defined by
\begin{align}
    s_V(X) &= \begin{cases}
        0 & X \subseteq V \\
        1 & \text{otherwise}
    \end{cases}
\end{align}
A function is {\em normal} if it is a non-negative linear combination of step functions.
Let $\Nor_n$ denote the set of normal functions on $[n]$.

To show the hardness result, we consider the dual linear programming formulation, which is
exactly $\Dsimple$. For easy reference we reproduce the LP below.
\begin{align*}
   && \min    & \qquad \sum_{i \in [k]} c_i \cdot \delta_i&& \label{eqn:DS} \\
         && \mbox{s.t.}  & \qquad \sum_{i \in [k]: \ X_i \subseteq V,  Y_i \not\subseteq V}
                        \delta_i \geq 1  & \forall V \subsetneq [n]   \nonumber
\end{align*}

Let $\po(\mathrm{DC})$ denote the convex region over $\delta$ defined by the constraints in
$\mathrm{DC}$. We first show the separation problem is hard.

\begin{lmm}
Given a degree constraint set $\mathrm{DC}$ and a    vector $\hat \delta \in \R_{\geq
0}^{|\mathrm{DC}|}$, checking if $\hat \delta \not\in \po(\mathrm{DC})$ is NP-complete.
Further, this remains the case under the extra condition that $\lambda \hat \delta  \in
\po(\mathrm{DC})$ for some $\lambda >1$.
\label{lmm:normal-separation}
\end{lmm}
We prove this lemma using a reduction from the Hitting Set problem, which is
well-known to be NP-complete. In the Hitting Set problem, the input is a set of $n$
elements $E = \{e_1,\dots, e_n\}$, a collection $\cS = \{S_1, \dots, S_m\}$ of $m$
subsets of $E$, and an integer $k > 0$. The answer is true iff there exists a subset $L$
of $k$ elements such that for every set $S_i \in \cS$ is `hit' by the set $L$ chosen,
i.e., $L\cap S_i \neq \emptyset$ for all $i \in [m]$.

Consider an arbitrary instance $H$ to the Hitting Set.  To reduce the problem to the
membership problem w.r.t. $\po(\mathrm{DC})$, we create an instance for computing the normal
bound that has $E' = E\cup \{e^*\}$ as variables and the following set $\mathrm{DC}$ of
degree constraints and $\hat \delta$ (here we do not specify the value of $c$ associated
with each degree constraint $(X, Y)$ as it can be arbitrary and we're concerned with the
hardness of the membership test).
\begin{enumerate}
\item $(\emptyset, \{e_i\})$ for all $e_i \in E$ with $\hat \delta_{\emptyset, \{e_i\}} = 1/(k+1)$.
\item $(S_i, E')$ for all $S_i \in \cS$ with $\hat \delta_{S_i, E'} = m$.
\item $(\{e^*\}, E')$ with $\hat \delta_{\{e^*\}, E'} = m$.
\end{enumerate}
Notably, to keep the notation transparent, we used $\hat \delta_{X, Y}$ to denote the value
of $\hat \delta$ associated with $(X, Y)$. Let $\mathrm{DC}_1$, $\mathrm{DC}_2$, and
$\mathrm{DC}_3$ denote the degree constraints defined above in each line respectively, and
let $\mathrm{DC} := \mathrm{DC}_1 \cup \mathrm{DC}_2 \cup \mathrm{DC}_3$.  To establish the
reduction we aim to show the following proposition.

\begin{prop}
There exists a hitting set of size $k$ in the original instance $H$ if and only if $\hat
\delta \not \in \po(\mathrm{DC})$.
\end{prop}
\begin{proof}
Let $L(V) := \sum_{(X, Y) \in \mathrm{DC}: X \subseteq V, V \not \subseteq Y} \hat
\delta_{Y|X}$. Let $L_{\min} := \min_{V\subset E'} L(V)$ and $V_{\min} := \arg \min_{V
\subset E'} L(V)$. To put the proposition in other words, we want to show that $H$ admits a
hitting set of size $k$ if and only if $L_{\min}  < 1$.

Let $\hat \delta(\mathrm{DC'}) := \sum_{(X, Y) \in \mathrm{DC'}} \hat \delta_{X, Y}$. Note
that $L_{\min} \leq L(\emptyset) = \hat \delta(\mathrm{DC}_1) = \frac{m}{k+1}$. Therefore,
we can have the following conclusions about $W_{\min}$.
\begin{itemize}
\item  $e^* \not \in V_{\min}$ since otherwise $L_{\min} \geq \hat \delta(\mathrm{DC})  =
m$.
\item For all $S_i \in \cS$, $S_i \not \subseteq V_{\min}$ since otherwise $L_{\min} \geq
\hat \delta(\{(S_i, E')\}) = m$.
\end{itemize}
Thus, we have shown that only the degree constraints in $\mathrm{DC}_1$ can contribute to
$L_{\min}$. As a result,
\[
L_{\min} = L(V_{\min}) = \hat \delta ( \{(\emptyset, \{e_i\}) \in \mathrm{DC}_1:  \{e_i\}
\not \subseteq V_{\min}\})  =   \frac{1}{k+1} |E \setminus V_{\min}|.
\]
As observed above,
for all $S_i \in \cS$, $S_i \not \subseteq V_{\min}$, which means $(E \setminus V_{\min})
\cap S_i \neq \emptyset$. This immediately implies that $E \setminus V_{\min}$ is a hitting
set.

To recap, if $\hat \delta \not \in \po(\mathrm{DC})$, we have $\frac{1}{k+1} |E \setminus
V_{\min}| < 1$ and therefore the original instance $H$ admits a hitting set $E \setminus
V_{\min}$ of size at most $k$.

Conversely, if the instance $H$ admits a hitting set $E'$ of size $k$, we can show that $L(E
\setminus E') = \hat \delta(\{(\emptyset, \{e_i\}) \in \mathrm{DC}_1 \; | \; e_i \in E'\}) =
\frac{k}{k+1} < 1$, which means  $\hat \delta \not \in \po(\mathrm{DC})$. This direction is
essentially identical and thus is omitted.
\end{proof}

The above proposition shows checking $\hat \delta \not \in \po(\mathrm{DC})$ is NP-hard. Further,
a violated constraint can be compactly represented by $V$; thus the problem is in NP.
Finally, if we scale up $\hat \delta$ by a factor of $\lambda = k+1$, we show $\lambda \hat
q \in \po(\mathrm{DC})$. We consider two cases. If $E \not \subseteq V$, we have $L(V) \geq
\hat \delta( \{ i \in [n] : e_i \not \in V \}) \geq \frac{1}{k+1} \lambda = 1$. If $E
\subseteq V$, it must be the case that $E = V$ since $V \neq E'$ and $E'= E \cup \{e^*\}$.
In this case $L(E) = \hat \delta(\mathrm{DC}_2) \geq m \lambda \geq 1$. Thus, for all $V
\subset E'$, we have $L(V) \geq 1$, meaning $\lambda \hat \delta \in \po(\mathrm{DC})$. This
completes the proof of Lemma~\ref{lmm:normal-separation}.

Using this lemma, we want to show that we can't solve $D'$ in polynomial time unless P =
NP. While there exist relationship among the optimization problem, membership problem and
their variants \cite{grotschel2012geometric}, in general hardness of the membership problem
doesn't necessarily imply hardness of the optimization problem. However, using the special
structure of the convex body in consideration, we can show such an implication in our
setting.
The following theorem would immediately imply Theorem~\ref{thm:normalhard}.

\begin{lmm}
We cannot solve $D'$ in polynomial time unless P = NP.
\end{lmm}
\begin{proof}
Consider an instance to the membership problem consisting of $\mathrm{DC}$ and $\hat
\delta$. By Lemma~\ref{lmm:normal-separation}, we know checking $\hat \delta \not \in
\po(\mathrm{DC})$ is NP-complete, even when $\lambda \hat \delta  \in  \po(\mathrm{DC})$ for
some $\lambda > 1$. For the sake of contradiction, suppose we can solve $D'_S$ in polynomial
time for any $w \geq 0$ over the constraints defined by the same $\po(\mathrm{DC})$. We will
draw a contradiction by showing how to exploit it to check $\hat \delta \not \in
\po(\mathrm{DC})$ in polynomial time.

Define $R := \{w \; | \; w \cdot (\delta - \hat \delta) > 0 \;\; \forall \delta \in
\po(G)\}$. 
It is straightforward to see that $R$ is convex.

We claim that $\hat \delta \not \in \po(\mathrm{DC})$ iff $R \neq \emptyset$. To show the
claim suppose $\hat \delta \not \in \po(\mathrm{DC})$. Recall from
Lemma~\ref{lmm:normal-separation} that there exists $\lambda > 1$ such that $\lambda \hat
\delta \in \po(\mathrm{DC})$. Let $\lambda' >0$ be the smallest $\lambda''$ such that
$\lambda'' \hat \delta \in \po(\mathrm{DC})$. Observe that $\lambda' > 1$ and $\lambda' \hat
\delta$ lies on a facet of $\po(\mathrm{DC})$, which corresponds to a hyperplane $\sum_{i
\in [k]: X_i \subseteq V, V\not \subseteq Y_i} \delta_i = 1$ for some $V \subset [n]$. Let
$w$ be the orthogonal binary vector of the hyperplane; so we have $w \cdot \lambda' \hat
\delta = 1$. Then, $w \cdot (\delta - \lambda' \hat \delta) \geq 0$ for all $\delta \in
\po(\mathrm{DC})$. Thus, for any $\delta \in \po(\mathrm{DC})$ we have $w \cdot  (\delta -
\hat \delta) \geq (\lambda'  -1) w \cdot  \hat \delta = \frac{\lambda' - 1}{\lambda} w \cdot
\lambda \hat \delta \geq \frac{\lambda' - 1}{\lambda} > 0$. The other direction is trivial
to show: If $\hat \delta \in \po(\mathrm{DC})$, no $w$ satisfies $w \cdot (\delta - \hat
\delta) > 0$ when $\delta = \hat \delta$.

Thanks to the claim, we can draw a contradiction if we can test if $R = \emptyset$ in
polynomial time. However, $R$ is defined on an open set which is difficult to handle.
Technically, $R$ is defined by infinitely many constraints but it is easy to see that we
only need to consider constraints for $\delta$ that are vertices of $\po(\mathrm{DC})$.
Further, $\po(\mathrm{DC})$ is defined by a finite number of (more exactly at most $2^n$)
constraints (one for each $W$). This implies that the following LP,
\begin{align*}
    \max & \;\epsilon \\
    w \cdot (\delta - \hat \delta) &\geq  \epsilon \quad  \forall \delta \in \po(\mathrm{DC}) \\
    w &\geq 0
\end{align*}
has a strictly positive optimum value iff $R \neq \emptyset$. We solve this using the
ellipsoid method. Here, the separation oracle is, given $w \geq 0$ and $\epsilon$, to
determine if $w \cdot (\delta - \hat \delta) \geq \epsilon$ for all $\delta \in
\po(\mathrm{DC})$; otherwise it should find a $\delta \in \po(\mathrm{DC})$ such that $w
\cdot (\delta - \hat \delta) < \epsilon$. In other words, we want to know $\min_{\delta \in
\po(\mathrm{DC})} w \cdot (\delta - \hat \delta)$.  If the value is no smaller than
$\epsilon$, all constraints are satisfied, otherwise, we can find a violated constraint,
which is given by the $\delta$ minimizing the value. But, because the oracle assumes $w
\cdot \hat \delta$ is fixed, so this optimization is essentially the same as solving  $D'$,
which can be solved by the hypothetical polynomial time algorithm we assumed to have for the
sake of contradiction. Thus, we have shown that we can decide in poly time if $R$ is empty
or not.
\end{proof}

\section{The flow bound}
\label{sec:flow:bound}
Before proving Theorem~\ref{thm:flow:bound}, we give a high-level overview.
The flow bound $ \flowbound(\dc, \pi)$ is based on a relaxation $\dc_\pi^{\sf flow}$  of the
input degree constraints $\dc$ that is less relaxed than the relaxation $\dc_\pi$ used in
the chain bound. In particular, every {\em simple} degree constraint in $\dc_\pi$ is
retained in $\dc_\pi^{\sf flow}$ as is.
Let $k_s$ denote the number of simple degree constraints.
And, for every non-simple degree constraint $(X,Y,c) \in \dc$,
$\dc_\pi^{\sf flow}$ contains the constraint $(X,Y',c)$ where $Y'\subseteq Y$ contains the
variables in $Y$ that come {\em after} all variables in $X$ in the permutation $\pi$ (note
that this is the same as $\dc_\pi$).
For convenience we reindex the degree constraints in
$\dc_\pi^{\sf flow}$ such that all $k_s$ simple degree constraints appear before any
non-simple degree constraints.

The flow bound $\flowbound(\dc,\pi)$ is then defined by the objective value of
the following polynomial-sized linear program:
\begin{align}
  D_{\sf flow}:
        &&\min \sum_{i \in [k]} c_i \cdot \delta_i  \\
  \textrm{s.t.}
        &&f_{i,t} &\leq \delta_i && \forall i \in [k_s], \forall t \in [n] \\
        &&\flow_t(t) &\geq 1 - \sum_{\substack{i\in [k]\setminus [k_s], \\ t \in Y_i - X_i}} \delta_i &&\forall t \in [n] & \label{eqn-less-excess} \\
        &&\flow_t(\emptyset) &\geq -1  && \forall t \in [n] \\
        &&\flow_t(Z) &\geq 0  && \forall Z   \in  G \setminus \{\emptyset\} \setminus \{t\}, \quad \forall t \in [n] \\
        && f_{i,t}&= 0  &&\forall i \in [k_s], t \in [n] \cap(Y_i \setminus X_i) &
\end{align}
where $\flow_t(Z)$ is  defined as follows:
\begin{align*}
   \flow_t(Z) &:=
     \sum_{i: Z = Y_i} f_{i,t} -
   \sum_{i: Z=X_i }f_{i,t} +
   \sum_{X: X\subset Z}\mu_{X,Z, t}+ \sum_{Y: Z\subset Y}\mu_{Z,Y, t}
\end{align*}
As in the dual linear program   $\Dflow$ for simple degree constraints, intuitively  $D_{\sf
flow}$ encodes   $n$ min-cost flow problems, however the difference is in the
constraint~\eqref{eqn-less-excess} that in  $D_{\sf flow}$  the demand of node  $t$ is
reduced (from 1) by an amount equal to the capacity of the non-simple degree constraints
that can route flow directly to $t$.
We additionally note that the objective only considers cardinality constraints, simple
degree constraints and non-simple constraints that agree with $\pi$.  Finally, flow is only
sent on simple degree constraints and cardinality constraints.

Note that one could modify the linear program for $\flowbound_\pi(\dc)$ by allowing the
``source'' for the flow to sink $t$ to not only be the empty set, but also any singleton
vertex in $[t-1]$. All of theoretical results would still hold for this modified linear
program, but this modified linear program would be better in practice as it would never
result in a worse bound, and for some instances it would result in a significantly better
bound. Further note that, by our reduction in Section~\ref{sec:lower-bound}, the
computing $\dc_\pi^{\sf flow}[\Gamma_n]$, the polymatroid bound on our relaxed degree
constraints $\dc_\pi^{\sf flow}$, is as hard as computing the polymatroid bound on arbitrary
instances.

\begin{proof}[Proof of Theorem~\ref{thm:flow:bound}]
  For part $(a)$, To show that $\dc[\Gamma_n] \le \flowbound(\dc, \pi)$ it is sufficient to
  show that a solution to the linear program for $\flowbound_\pi(\dc)$ can be extended to a
  feasible flow for linear program $D$ as we did for simple degree constraints in
  section~\ref{subsec:simple}.  The only difference  here is that  some flow can be directly
  pushed to a sink $t$ on nonsimple edges in $\dc_t$. The fact that the flow bound is smaller
  than the chain bound follows immediately from the fact that the degree constraints used in
  the chain bound are a subset of the degree constraints used in the flow bound.

  Statement $(b)$ follows because for simple instances the linear program for $\flowbound_\pi$
  and $D_S'$ are identical, and for acyclic instances the polymatroid bound equals the chain
  bound~\cite{DBLP:conf/pods/000118}.

  Finally, we prove part $(c)$, stating the chain bound can be arbitrarily larger than the
  flow bound follows from the following instance.  Consider an instance consisting of two
  elements $1$ and $2$ and let the permutation $\pi$ follow this order.  There is a
  cardinality constraint $(\emptyset, \{2\},1)$ and simple degree constraint $(\{2\}, \{1,2\},
  1)$. The chain bound   is unbounded because it cannot use the simple degree constraint that
  does not agree with $\pi$. Alternatively, the flow bound is bounded by using both degree
  constraints.
\end{proof}

\section{Concluding Remarks}

Our main contributions are  polynomial-time algorithms to compute the polymatroid bound and
polynomial length proof sequences for simple degree constraints. These results nudge the
information theoretic framework from~\cite{csma,panda} towards greater practicality. In
fact, our technique and the flow-bound from Section~\ref{sec:flow:bound} were adopted in the
recent work of Zhang et al.~\cite{DBLP:journals/corr/abs-2502-05912} to make part of their
cardinality estimation framework practical.

The main major open problem  remains determining the computational complexity of the
polymatroid bound. While we proved some negative results regarding the hardness of computing
the polymatroid bound beyond simple degree constraints, we should still be looking for other
ways to parameterize the input so that the polymatroid bound can be computed in polynomial
time.

\newcommand{\grantsponsor}[3]{\textsc{#1}}
\newcommand{\grantnum}[3]{\textsc{#3}}

\section*{Acknowledgments}
  Sungjin Im was supported in part by  \grantsponsor{1}{NSF}{https://www.nsf.gov/} grants
  CCF-\grantnum{https://www.nsf.gov/awardsearch/showAward?AWD_ID=2423106}{}{2423106},
  CCF-\grantnum{https://www.nsf.gov/awardsearch/showAward?AWD_ID=2121745}{}{2121745}, and
  CCF-\grantnum{https://www.nsf.gov/awardsearch/showAward?AWD_ID=1844939}{}{1844939}, and
  \grantsponsor{2}{Office of Naval Research}{https://www.onr.navy.mil/} Award
  \grantnum{}{}{N00014-22-1-2701}. Benjamin Moseley was supported in part by a Google
  Research Award, an Infor Research Award, a Carnegie Bosch Junior Faculty Chair,
  \grantsponsor{1}{NSF}{https://www.nsf.gov/} grants
  CCF-\grantnum{https://www.nsf.gov/awardsearch/showAward?AWD_ID=2121744}{}{2121744} and
  CCF-\grantnum{https://www.nsf.gov/awardsearch/showAward?AWD_ID=1845146}{}{1845146}, and
  \grantsponsor{2}{Office of Naval Research}{https://www.onr.navy.mil/} Award
  \grantnum{}{}{N00014-22-1-2702}. Kirk Pruhs was supported in part by the
  \grantsponsor{1}{NSF}{https://www.nsf.gov/} grants
  CCF-\grantnum{https://www.nsf.gov/awardsearch/showAward?AWD_ID=2209654}{}{2209654} and
  CCF-\grantnum{https://www.nsf.gov/awardsearch/showAward?AWD_ID=1907673}{}{1907673},  and
  an IBM Faculty Award. Part of this work was conducted while the authors participated in
  the Fall 2023 Simons Program on Logic and Algorithms in Databases and AI.

\bibliographystyle{siam}
\bibliography{main}

\end{document}